\documentclass{article} 
\usepackage{arxiv}  
\usepackage[utf8]{inputenc} 
\usepackage[T1]{fontenc}    
\usepackage{hyperref}       
\usepackage{url}            
\usepackage{booktabs}       
\usepackage{amsfonts}       
\usepackage{nicefrac}       
\usepackage{microtype}      
\usepackage{lipsum}

\usepackage{times}
\usepackage{soul}
\usepackage{url}
\usepackage[small]{caption}
\usepackage{graphicx}
\usepackage{amsmath}
\usepackage{amsthm}
\usepackage{booktabs}
\usepackage{algorithm}
\usepackage{todonotes}
\usepackage{mathtools}
\usepackage{multirow}

\usepackage{amssymb}
\usepackage{natbib}
\usepackage{hyperref}

\usepackage{amsfonts}
\usepackage{algorithm}
\usepackage[noend]{algpseudocode}
\usepackage{mathtools}
\usepackage{amsthm}
\usepackage{thmtools}
\usepackage{thm-restate}
\usepackage{lipsum}
\usepackage{amssymb }
\usepackage{mleftright}
\usepackage{color}
\usepackage{diagbox}
\usepackage{tikz}
\usepackage{wrapfig}
\usetikzlibrary{positioning,fit,calc,matrix,arrows,trees,backgrounds}
\usepackage{accents}
\usepackage{tikz}
\usepackage{pgfplots}
\usepackage{enumerate}
\usepackage[shortlabels]{enumitem}
\usepackage{svg}
\usepackage{subcaption}
\usepackage{etoolbox}

\newcommand{\I}{\mathcal{I}}

\newcommand{\G}{\mathcal{G}}

\newcommand{\Z}{\mathcal{Z}}

\newcommand{\C}{\mathcal{C}}

\newcommand{\Hi}{\mathcal{H}}

\newcommand{\onevec}{\boldsymbol{1}}

\newcommand{\defeq}{\coloneqq}
\newcommand{\Reals}{\mathbb{R}}

\newcommand{\xvec}{\boldsymbol{x}}
\newcommand{\gvec}{\boldsymbol{g}}
\newcommand{\wvec}{\boldsymbol{w}}

\newcommand{\yvec}{\boldsymbol{y}}
\newcommand{\zvec}{\boldsymbol{z}}

\newcommand{\zerovec}{\boldsymbol{0}}
\newcommand{\Fmat}{\boldsymbol{F}}
\newcommand{\Mmat}{\boldsymbol{M}}

\newcommand{\fvec}{\boldsymbol{f}}
\newcommand{\mvec}{\boldsymbol{m}}
\newcommand{\Y}{\mathcal{Y}}
\newcommand{\X}{\mathcal{X}}
\newcommand{\Umat}{\boldsymbol{U}}

\newcommand{\Amat}{\boldsymbol{A}}
\newcommand{\vvec}{\boldsymbol{v}}

\newtheorem{lemma}{Lemma}

\newtheorem{corollary}{Corollary}

\newtheorem{definition}{Definition}

\usepackage{pifont}
\usepackage[normalem]{ulem} 

\providetoggle{comment}
\settoggle{comment}{true} 

\definecolor{mygreen}{rgb}{0.0, 0.5, 0.0}
\definecolor{colormatri}{rgb}{1, 0.38, 0}

\allowdisplaybreaks

\newcommand{\mycircled}[1]{%
	\begin{tikzpicture}[baseline={(char.base)}]
		\node[draw,circle,inner sep=0.5pt] (char){#1};
	\end{tikzpicture}
}

\makeatletter
\def\hlinewd#1{%
	\noalign{\ifnum0=`}\fi\hrule \@height #1 \futurelet
	\reserved@a\@xhline}
\makeatother

\errorcontextlines\maxdimen
\usepackage{etoolbox}
\usepackage{tikz}
\usetikzlibrary{tikzmark}
\usetikzlibrary{calc}
\newcommand{\ALGtikzmarkcolor}{black}
\newcommand{\ALGtikzmarkextraindent}{4pt}
\newcommand{\ALGtikzmarkverticaloffsetstart}{-.5ex}
\newcommand{\ALGtikzmarkverticaloffsetend}{-.5ex}
\makeatletter
\newcounter{ALG@tikzmark@tempcnta}

\newcommand\ALG@tikzmark@start{%
	\global\let\ALG@tikzmark@last\ALG@tikzmark@starttext%
	\expandafter\edef\csname ALG@tikzmark@\theALG@nested\endcsname{\theALG@tikzmark@tempcnta}%
	\tikzmark{ALG@tikzmark@start@\csname ALG@tikzmark@\theALG@nested\endcsname}%
	\addtocounter{ALG@tikzmark@tempcnta}{1}%
}

\def\ALG@tikzmark@starttext{start}
\newcommand\ALG@tikzmark@end{%
	\ifx\ALG@tikzmark@last\ALG@tikzmark@starttext
	\else
	\tikzmark{ALG@tikzmark@end@\csname ALG@tikzmark@\theALG@nested\endcsname}%
	\tikz[overlay,remember picture] \draw[\ALGtikzmarkcolor] let \p{S}=($(pic cs:ALG@tikzmark@start@\csname ALG@tikzmark@\theALG@nested\endcsname)+(\ALGtikzmarkextraindent,\ALGtikzmarkverticaloffsetstart)$), \p{E}=($(pic cs:ALG@tikzmark@end@\csname ALG@tikzmark@\theALG@nested\endcsname)+(\ALGtikzmarkextraindent,\ALGtikzmarkverticaloffsetend)$) in (\x{S},\y{S})--(\x{S},\y{E});%
	\fi
	\gdef\ALG@tikzmark@last{end}%
}

\apptocmd{\ALG@beginblock}{\ALG@tikzmark@start}{}{\errmessage{failed to patch}}
\pretocmd{\ALG@endblock}{\ALG@tikzmark@end}{}{\errmessage{failed to patch}}
\makeatother

\usepackage{newfloat}
\usepackage{listings}
\DeclareCaptionStyle{ruled}{labelfont=normalfont,labelsep=colon,strut=off} 
\floatstyle{ruled}
\newfloat{listing}{tb}{lst}{}
\floatname{listing}{Listing}
\title{Last-iterate Convergence to Trembling-hand Perfect Equilibria}
\author{
	Martino Bernasconi\\
	Politecnico di Milano\\
	\texttt{martino.bernasconideluca@polimi.it}
	\And
	Alberto Marchesi\\
	Politecnico di Milano\\
	\texttt{alberto.marchesi@polimi.it}
	\And
	Francesco Trovò\\
	Politecnico di Milano\\
	\texttt{francesco1.trovo@polimi.it}
}

\begin{document}

\maketitle

\begin{abstract}
	Designing efficient algorithms to find \emph{Nash equilibrium} (NE) refinements in sequential games is of paramount importance in practice.
	Indeed, it is well known that the NE has several weaknesses, since it may prescribe to play sub-optimal actions in those parts of the game that are never reached at the equilibrium.
	NE refinements, such as the \emph{extensive-form perfect equilibrium} (EFPE), amend such weaknesses by accounting for the possibility that players may make mistakes.
	This is crucial in real-world applications, where humans having bounded rationality are usually involved, and it turns out being useful also in boosting the performances of superhuman agents for recreational games like Poker.
	Nevertheless, only few works addressed the problem of computing NE refinements.
	Most of them propose algorithms finding exact NE refinements by means of linear programming, and, thus, these do \emph{not} have the potential of scaling up to real-world-size games.
	On the other hand, existing iterative algorithms that exploit the tree structure of sequential games only provide convergence guarantees to {approximate} refinements.
	In this paper, we provide the first efficient \emph{last-iterate} algorithm that provably converges to an EFPE in two-player zero-sum sequential games with imperfect information.
	Our algorithm works by tracking a sequence of equilibria of suitably-defined, regularized-perturbed games.
	In order to do that, it uses a procedure that is specifically tailored to converge last iterate to the equilibria of such games.
	Crucially, the updates performed by such a procedure can be performed efficiently by visiting the game tree, thus rendering our algorithm more appealing than its linear-programming-based competitors in terms of scalability.
	Finally, we evaluate our algorithm on a standard testbed of games, showing that it outputs strategies which are much more robust to players' mistakes than those of state-of-the-art NE-computation algorithms.
\end{abstract}
\section{Introduction}

Computing the \emph{Nash equilibria} (NEs)~\citep{nash1951non} of \emph{two-player} \emph{zero-sum} sequential (\emph{i.e.}, extensive-form) games with \emph{imperfect information} has been one of the flagship computational challenges of artificial intelligence since several years.
The latest advances in such a research field have lead to the development of \emph{superhuman} agents that are capable of beating top human professionals in several recreational games, such as, \emph{e.g.}, Go~\citep{gocitation} and heads-up no-limit Texas hold'em Poker~\citep{Brown418,brown2019superhuman}.

The NE is a solution concept that prescribes each player to play optimally under the assumption that the opponents are perfectly rational and do the same.
However, this assumption is oftentimes unreasonable, especially when computed equilibria are deployed in the real world, where artificial agents usually face human opponents that naturally have bounded rationality.
Moreover, it is well known that the NE has several weaknesses when played against opponents with bounded rationality who may make mistakes.
In particular, an NE may prescribe to perform sub-optimal actions at decision points (a.k.a. \emph{information sets}) that are never reached assuming the players play equilibrium strategies.

Over the last decades, game theorists introduced several \emph{refinements} of the NE notion, in order to amend its weaknesses off the equilibrium path.
Among them, the most studied one is the \emph{extensive-form perfect equilibrium} (EFPE) introduced by~\citet{selten1975reexamination}.
The EFPE is based on the idea of \emph{trembling-hand perfection}, whose rationale is to let the players reasoning about the possibility that both themselves and their opponents may ``tremble'' in future, by playing sub-optimal actions with small probability at information sets that will be reached later during the game.
Then, an EFPE is defined as a limit point of a sequence of equilibria that is obtained by letting the ``magnitude of trembles'' (\emph{i.e.}, the probabilities of playing sub-optimal actions at future information sets) going to zero.
This makes the strategies prescribed by an EFPE robust against possible players' mistakes.

The problem of computing refinements of the NE has received some attention from the artificial intelligence community only recently.
This is surprising, since designing artificial agents that are capable of exploiting opponents' mistakes is of paramount importance when bringing equilibrium concepts into practice.
Indeed, the adoption of NE refinements would foster the operationalization of equilibrium computation techniques into novel application scenarios where ``humans are in the loop'', such as, \emph{e.g.}, military settings and businesses.
At the same time, NE refinements could boost the performances of state-of-the-art superhuman agents playing recreational games, by equipping them with ability of capitalizing over opponents' mistakes.

Despite there are several appealing reasons for switching the attention from the computation of NEs to that of its refinements, only few works addressed the latter problem in two-player zero-sum sequential games with imperfect information.
Moreover, most of these works, such as, \emph{e.g.},~\citep{farina2017extensive,farina2018practical,farina2021equilibrium}, provide equilibrium computation algorithms that rely on the solution of \emph{linear programs} (LPs).
As a result, up to now, these methods have failed at scaling up on huge games like Texas hold'em Poker endgames.
This was predictable, as the major breakthroughs in NE computation in two-player zero-sum games (such as, \emph{e.g.},~\citep{Brown418,brown2019superhuman}) were achieved by means of iterative techniques that work by repeatedly visiting the game tree (see, \emph{e.g.}, the \emph{counterfactual regret minimization} (CFR) framework by~\citet{zinkevich2007regret}).
To the best of our knowledge, the only iterative methods for computing NE refinements are those presented in~\citep{farina2017regret}~and~\citep{kroer2017smoothing}.
However, these do \emph{not} provide any convergence guarantee to exact refinements, but only to approximate equilibria of a perturbed game (see the following Section~\ref{sec:related} for more details).

In this paper, we provide the first \emph{iterative} algorithm that provably converges to (exact) NE refinements, focusing on the case of EFPEs.
Our algorithm works by tracking a \emph{sequence of equilibria} of suitably-defined, regularized-perturbed games, which is different from the sequence in the definition of EFPE.
Indeed, these game modify the utility function of the original game by adding particular \emph{regularization} and \emph{perturbation} components, while the definition of EFPE only considers the latter.
Intuitively, the regularization component is needed to ensure equilibrium uniqueness, while the perturbation component guarantees that the equilibrium strategies place at least a given (small) probability on every action at every information set, resembling the idea of trembling-hand perfection.
We prove that, by carefully controlling how the regularization and the perturbation components vanish, the resulting sequence of equilibria admits an EFPE of the original game as a limit point.
In particular, this is achieved by letting the regularization component vanishing faster than the perturbation one.

The core ingredient of our algorithm is a \emph{last-iterate} procedure that provably converges to an (approximate) equilibrium of a given regularized-perturbed game.
Such a procedure extends the \emph{optimistic online mirror descent} (OOMD) algorithm by~\citet{rakhlin2013optimization}.
Intuitively, our procedure uses an OOMD-style update rule in order to converge last iterate, while it deals with the non-smooth terms appearing in the utility of regularized-perturbed games without resorting to linearization techniques.
The crucial feature of our sub-procedure is that its updates can be performed by visiting the game tree recursively, thus avoiding the solution of any ``one-piece'' optimization problem.
This considerably enhances the scalability of our algorithm compared to those relying on the solution of optimization problems that are directly cast on players' strategy spaces, like LP-based ones.

We conclude by experimentally evaluating the performances of our algorithm on a standard testbed of Poker-inspired game instances.
Our analysis shows that our algorithm outperforms of orders of magnitude several state-of-the-art NE-computation algorithms (such as, \emph{e.g.}, the CFR algorithm) in terms of average players' regrets over all the information sets of the game, which is the standard practical metric used for evaluating convergence to NE refinements.

\section{Related Works}\label{sec:related}

In this section, we review the state of the art on the computation of EFPEs and other NE refinements.
See the book by~\citet{van1991stability} for a detailed treatment of the latter.

The first works to address the computational problem of finding NE refinements in two-player zero-sum sequential games are~\cite{miltersen2010computing}~and~\cite{farina2017extensive}, which focus on the \emph{quasi-perfect equilibrium} (QPE) and the EFPE, respectively.
The algorithms presented in these papers work by solving a ``perturbed'' LP which is a modified version of the classical LP for finding an NE in two-player zero-sum games.
In particular, the LP is ``perturbed'' in the sense that players' strategies are constrained to placing at least some small probability on every action at every information set, so as to resemble the idea of trembling-hand perfection.
The authors of these papers show that, if the probability lower bound is sufficiently small, then starting from a solution to a ``perturbed'' LP one can recover an EFPE or a QPE exactly.
However, such a lower bound must be prohibitively small in general, and this makes the ``perturbed'' LP unsolvable by means of standard solvers, requiring the adoption of infinite-precision ones.
As a result, these methods are highly unpractical.

\citet{farina2018practical} made a step further by defining a general theory of \emph{trembling} LPs, which are LPs parametrized by some parameter $\epsilon$, and proposing a polynomial-time algorithm that finds a limit solution to a sequence of trembling LPs as $\epsilon \to 0$.
Such a framework allows to encompass several NE refinements, including EFPE, QPE, and some newly-introduced variations of them~\citep{farina2021equilibrium}.
Moreover, the algorithm proposed by~\citet{farina2018practical} requires the solution of polynomially-many (standard) LPs, whose coefficients are \emph{not} prohibitively small as those of ``perturbed'' LPs.
This allows for the use of standard solvers, which render the algorithm deployable in practice.
However, such method does \emph{not} exploit the tree structure of the game, and this prevents it from scaling up to real-world-size game instances.

The only algorithms that compute EFPEs by iteratively visiting the game tree are those in~\citep{farina2017regret}~and~\citep{kroer2017smoothing}.
The former introduces a modification of the CFR algorithm that allows it to take into account opponents' mistakes so as to minimize regret at all the information sets, including those that are never reached at the equilibrium.
The latter applies a similar idea to first-order methods for sequential games.
In this way, such iterative methods are able to converge to an approximate equilibrium of a perturbed game.
However, they are \emph{not} concerned with the computation of a limit point of a sequence of equilibria, since they work by approximating an equilibrium while keeping the ``magnitude of trembles'' fixed.
As a result, they do \emph{not} provide any theoretical convergence guarantee to an exact EFPE, which would require taking a limit point as ``trembles'' vanish.

In conclusion, it is also worth citing some works that study the problem of computing NE refinements in $n$-player general-sum games; see, \emph{e.g.},~\citep{miltersen2010computing,hansen2010computational,farina2017extensive,gatti2020characterization}.

\section{Preliminaries}

In this section, we provide all the definitions and concepts needed in the paper.
First, we formally introduce \emph{two-player zero-sum extensive-form games with imperfect information} (EFGs for short).
Then, we define the sequence-form representation for players' strategies and the notion of EFPE.  

\subsection{Extensive-form Games}

An EFG is usually described by means of a game tree.
We denote by $\Hi$ the set of nodes of the tree and by $\Z\subset \Hi$ the set of terminal nodes, which are the leaves of the tree.
Each non-terminal node $h \in \Hi \setminus \Z$ is either a decision node in which a given player acts or a chance node where a random event occurs.
For any decision node $h$, we let $A(h)$ be the set of actions available to the player acting at $h$.
On the other hand, each terminal node $z\in Z$ is associated with a payoff $u(z)\in[-1,1]$ for the first player, while $-u(z)$ is the payoff of the second player since the game is zero sum.

In an EFG, imperfect information is usually described by means of \emph{information sets} (infosets for short).
A player's infoset $I$ is a collection of decision nodes of that player which are indistinguishable for them; formally, it must be the case that $A(h)=A(h^\prime)$ for any pair $h,h^\prime\in I$.
We let $A(I)$ be the set of actions available at all nodes of infoset $I$, with $n_I \coloneqq|A(I)|$ being its cardinality. 
Moreover, we denote by $\I_1$ and $\I_2$ the sets all the infosets of player $1$ and player $2$, respectively.
As customary in the literature, we restrict the attention to games with \emph{perfect recall}, where a player never forgets information once acquired.
This is equivalent to assuming that both $\I_1$ and $\I_2$ are partially ordered according to a suitably-defined precedence relation.
We denote by $\preceq$ such a relation between infosets, where $I \preceq J$ means that infoset $I$ \emph{precedes} infoset $J$.
Formally, $I \preceq J$ if and only if there exists a path in the game tree that connects a node belonging to infoset $I$ to a node in infoset $J$.
Furthermore, give any infoset $I \in \I_i$ of player $i$ and action $a \in A(I)$, we let $C_{I,a}$ be the set of all player $i$'s infosets that \emph{immediately follow $I$ through action $a$}, according to the relation $\preceq$.

\subsection{Sequence-form Representation}

Any node $h \in \Hi$ defines a \emph{sequence} $\sigma_i(h)$ of player $i$'s actions encountered on the path from the root of the game tree to $h$.
In EFGs with perfect recall, an infoset $I\in \I_i$ uniquely determines a {sequence} of player $i$'s actions, since $\sigma_i(h) = \sigma_i(h')$ for any $h, h' \in I$ by definition.
We denote such a sequence by $\sigma_i(I)$.
Since $\sigma_i(I)$ extended with any action $a \in A(I)$ is a valid sequence of player $i$'s actions, we can identify player $i$'s sequences with infoset-action pairs.
Thus, we let $\Sigma_i \coloneqq \{(I,a) \mid  I\in \I_i, a\in A(I)\}\cup\{\varnothing\}$ be the set of player $i$'s sequences, where $\varnothing$ is the empty sequence defined by paths in the tree in which player $i$ never plays.

By leveraging the sequence form, the \emph{mixed strategies} of a player can be encoded in terms of realization probabilities of sequences~\citep{von1996efficient}.  
Formally, a first player's strategy is a vector $\xvec\in[0,1]^{|\Sigma_1|}$ such that, for each $\sigma\in\Sigma_1$, the entry $\xvec[\sigma]$ is the probability of playing sequence $\sigma$.
To be well defined, $\xvec$ must satisfy the following linear constraints:
\[
\textstyle{
\xvec[\varnothing]=1 \,\,\, \text{and} \,\,\, \xvec[\sigma_i(I)] = \sum_{a\in A(I)} \xvec[\sigma_i(I)a] \quad  \forall I \in \I_1,
}
\]
which can be expressed as $\Fmat_1 \xvec = \fvec_1$ using matrix notation.
Similarly, we let $\yvec\in[0,1]^{|\Sigma_2|}$ be a strategy for the second player, which must satisfy the constraints $\Fmat_2 \yvec = \fvec_2$.
In the following, we denote by $\X$ and $\Y$ the polytopes of valid strategies for the first and the second player, respectively.

Thanks to the sequence-form strategy representation, we can define the first player's expected utility given two strategies $\xvec \in \X$ and $\yvec \in \Y$ as the bilinear term $\xvec^\top \Umat \yvec$, where $\Umat \in [-1,1]^{|\Sigma_1| \times |\Sigma_2|}$ is a utility matrix whose entry corresponding to $\sigma_1 \in \Sigma_1$ and $\sigma_2 \in \Sigma_2$ is defined as follows:
\[
\Umat[\sigma_1,\sigma_2] \coloneqq \sum\limits_{\substack{z\in\Z: \sigma_1(z)=\sigma_1 \wedge \sigma_2(z)=\sigma_2}} p(z)u(z),
\]
with $p(z) \in [0,1]$ being the product of probabilities associated to chance outcomes on the path from the root to $z\in\Z$.
The second player's expected utility is $-\xvec^\top \Umat \yvec$. 
\subsection{Extensive-form Perfect Equilibria}

The EFPE refines the NE by considering the possibility that players may make mistakes and play off-equilibrium actions with ``small and vanishing'' probability. 
Formally, an EFPE is defined as a limit point as $\epsilon \to 0$ of a sequence of NEs of ``perturbed games'' parametrized by $\epsilon > 0$, where each action $a \in A(I)$ at each infoset $I$ must be played with probability at least $\epsilon$~\citep{selten1975reexamination}.
In terms of sequence-form strategies, this is equivalent to ask, for the first player, that $\xvec[\sigma_i(I)a]\ge\epsilon \, \xvec[\sigma_i(I)]$ for every $I \in \I_1$ and $a \in A(I)$.
Such linear constraints can be expressed as a polytope of the form $\Mmat_1(\epsilon)\xvec\ge \mvec_1(\epsilon)$ (see~\citet{farina2017extensive}), so that we can define the set of valid first player's strategies for a perturbed game parametrized by $\epsilon > 0$ as 
\[
\X^{\epsilon} \coloneqq \{\xvec\in \X\mid \Fmat_1\xvec = \fvec_1, \Mmat_1(\epsilon)\xvec\ge \mvec_1(\epsilon)\}.
\]
Similarly, for the second player we have:
\[
\Y^{\epsilon} \coloneqq \{\yvec\in \Y\mid \Fmat_2\yvec = \fvec_2, \Mmat_2(\epsilon)\yvec\ge \mvec_2(\epsilon)\}.
\]

Then, we can provide the following definition of EFPE:\footnote{For ease of notation, $\epsilon\to0$ denotes the limit from the right.}
\begin{definition}[EFPE]\label{def:eps_efpe}
	Given any $\epsilon>0$, an \emph{$\epsilon$-EFPE} is defined as any pair of strategies $(\xvec^\star,\yvec^\star)\in\X^{\epsilon}\times\Y^{\epsilon}$ such that, for all $\xvec\in\X^{\epsilon}$ and $\yvec \in \Y^{\epsilon}$, it holds:
	\[
	\xvec^\top \Umat \yvec^\star \le \xvec^{\star,\top} \Umat \yvec^\star \le \xvec^{\star, \top} \Umat \yvec.
	\]
	Moreover, an \emph{EFPE} is a limit point of $\epsilon$-EFPEs as $\epsilon\to0$.
\end{definition}

In this work, we also introduce a relaxed notion of EFPE, which we call $\delta$-approximate $\epsilon$-EFPE and it is defined as $\delta$-approximate NEs of $\epsilon$-perturbed games.
Formally:
\begin{definition}\label{def:delta_eps_efpe}
	Given any $\epsilon>0$ and $\delta>0$, a \emph{$\delta$-approximate $\epsilon$-EFPE} is a pair of strategies $(\xvec^\star,\yvec^\star)\in\X^{\epsilon}\times\Y^{\epsilon}$ such that, for all $\xvec\in\X^{\epsilon}$ and $\yvec \in \Y^{\epsilon}$, it holds:
	\[
	\xvec^\top \Umat \yvec^\star-\delta \le \xvec^{\star,\top} \Umat \yvec^\star \le \xvec^{\star,\top} \Umat \yvec+\delta.
	\]
\end{definition}
Notice that, when $\delta=0$ and by letting $\epsilon\to0$, the definition of $\delta$-approximate $\epsilon$-EFPE coincides with that of EFPE.

Finally, in order to measure how well a pair of players' strategies approximates an NE, we need to introduce the \emph{Nash gap}.
For every $(\tilde\xvec,\tilde\yvec)\in\X\times\Y$, this is formally defined as $Gap(\tilde\xvec,\tilde\yvec):=\max_{\xvec\in\X}\xvec^\top\Umat\tilde\yvec-\min_{\yvec\in\Y}\tilde{\xvec}^\top\Umat\yvec$.

\section{Sequences of Equilibria Leading to EFPEs}\label{sec:methods}
In this section, we provide the core results that allow us to design our algorithm converging to an EFPE.
First, we introduce a modified version of an EFG, which alters the utility function of the game by adding suitable \emph{regularization} and \emph{perturbation} components, where the former guarantees equilibrium uniqueness, while the latter ensures that players' strategies at the equilibrium are valid for the ``perturbed games'' in the definition of EFPE.
Then, we show that such a modified game allows us to identify a \emph{sequence of equilibria} that admits an EFPE as a limit point, by carefully tuning the parameters controlling the regularization and the perturbation components in the modified game.
In the following Section~\ref{sec:algo}, we provide an algorithm that is able to track such a sequence of equilibria, thus converging to an EFPE.

\subsection{Regularized and Perturbed Games}

Next, we formally define regularized-perturbed games, and we show some of their properties that will be useful in proving our core results, namely equilibrium uniqueness and the connection between their equilibria and the $\delta$-approximate $\epsilon$-EFPEs of the original game (Theorem~\ref{thm:connection_efpe}).

Given any EFG and two parameters $\lambda,\epsilon > 0$, we define the \emph{regularized-perturbed game} $\G(\lambda,\epsilon)$ as a two-player zero-sum game in which the players' strategy sets are those of the original game (namely $\X$ and $\Y$) and the first player's utility for any pair $\xvec \in \X$ and $\yvec \in \Y$ is given by:
\begin{equation*}
f_{\lambda,\epsilon}(\xvec,\yvec) \coloneqq \xvec^\top \Umat\yvec - \frac{1}{\lambda} d_1^\epsilon(\xvec) + \frac{1}{\lambda} d_2^\epsilon(\yvec),
\end{equation*}
where $d_1^\epsilon:\mathcal{X}\to\mathbb{R}$ and $d_2^\epsilon:\mathcal{Y}\to\mathbb{R}$ are strongly convex functions that are defined recursively over the sequence-form strategy sets $\X$ and $\Y$, respectively, as follows:
\begin{align*}
d^\epsilon_1(\xvec) \coloneqq \sum\limits_{I\in\I_1} \alpha_I \, \xvec[\sigma_1(I)]d_{\Delta_I}^{\epsilon}\left(\frac{\xvec[I]}{\xvec[\sigma_1(I)]}\right),\quad d^\epsilon_2(\yvec) \coloneqq \sum\limits_{I\in\I_2} \alpha_I \, \yvec[\sigma_2(I)]d_{\Delta_I}^{\epsilon}\left(\frac{\yvec[I]}{\yvec[\sigma_2(I)]}\right).
\end{align*}

In the definitions of $d^\epsilon_1(\xvec)$ and $d^\epsilon_2(\yvec)$, for ease of presentation, we introduced the following useful additional notation:
\begin{itemize}
	\item for every infoset $I \in \mathcal{I}_1$ (respectively $I \in \mathcal{I}_2$), the vector $\xvec[I] \in [0,1]^{n_I}$ (respectively $\yvec[I] \in [0,1]^{n_I}$) is the sub-vector of $\xvec$ (respectively $\yvec$) made by all the components $\xvec[\sigma_1(I)a]$ (respectively $\yvec[\sigma_2(I) a]$) for $a \in A(I)$; 
	\item for every $I \in \mathcal{I}_1 \cup \mathcal{I}_2$, the function $d_{\Delta_I}^{\epsilon}:\Delta^{n_I}\to\mathbb{R}$ is such that $d_{\Delta_I}^{\epsilon}(\wvec) \coloneqq (\wvec-\onevec\epsilon)^\top\log(\wvec-\onevec\epsilon)$;\footnote{In this work, $\Delta_I$ denotes the $(n_I-1)$-dimensional simplex defined over the set $A(I)$ of actions at infoset $I \in \I_1 \cup \I_2$. We also use $\log(\wvec)$ as a shorthand for the vector whose $k$-th component is $\log(\wvec[k])$. Moreover, for the function $d^\epsilon_{\Delta_I}$ to be well defined for every $I \in \I_1 \cup \I_2$, we assume w.l.o.g. that $\epsilon\le \min_{I\in\I_1\cup\I_2}{1}/{2n_I}$.}
	\item for every $I \in \mathcal{I}_1 \cup \mathcal{I}_2$, the weights $\alpha_I \in \mathbb{R}_+$ are recursively defined as
	$
	\alpha_I \defeq 2+2\max_{a\in A(I)}\sum_{J\in \C_{I,a}}\alpha_{J},
	$
	so as to guarantee that $d_1^\epsilon$ and $d_2^\epsilon$ are $1$-strongly convex functions w.r.t. the $\ell_2$-norm (see Lemma~\ref{lem:strongconvexity} in Appendix~\ref{app:1}).
	
\end{itemize}

In the following, we will also introduce the \emph{Bregman divergences} $D_i^\epsilon(\cdot|\cdot)$ associated with the distance-generating functions $d_i^\epsilon(\cdot)$.
Formally, for $\xvec, \tilde\xvec \in \X$ we define:
\[
D_1^\epsilon(\xvec|\tilde\xvec)\coloneqq d_1^\epsilon(\xvec)-d_i^\epsilon(\tilde\xvec)-\nabla d_i^\epsilon(\tilde\xvec)^\top(\xvec-\tilde\xvec),
\]
and $D_2^\epsilon(\yvec|\tilde\yvec)$ is defined analogously.

Notice that the functions $d_1^\epsilon$ and $d_2^\epsilon$ are special cases of \emph{dilated} distance-generating functions, which are defined by directly exploiting the tree form of sequence-form strategy sets $\X$ and $\Y$~\citep{hoda2010smoothing,lee2021last}.
In particular, $d_1^\epsilon$ and $d_2^\epsilon$ employ the $d^\epsilon_{\Delta_I}$ as base distance-generating functions for the simplexes defining the strategy spaces at every infoset $I$.
These functions modify the classical negative entropy by offsetting by $-\epsilon$ the strategy given as input.
Intuitively, this is needed in order to ensure that the equilibria of the game belong to $\X^{\epsilon} \times \Y^{\epsilon}$, and, thus, the parameter $\epsilon$ can be used to tune the perturbation component.

We also notice that the term $1/\lambda$ that multiplies the functions $d_1^\epsilon$ and $d_2^\epsilon$ is the one controlling the regularization component, and it is crucial to recover the strong convexity-concavity of the utility function $f_{\lambda,\epsilon}$, which, in its turn, guarantees equilibrium uniqueness.
By letting $\zvec_{\lambda,\epsilon}^\star \coloneqq (\xvec_{\lambda,\epsilon}^\star, \yvec_{\lambda,\epsilon}^\star) \in \X^\epsilon \times \Y^\epsilon$ be an NE of the game $\G(\lambda,\epsilon)$, it is easy to check that it must satisfy the following conditions:
\begin{align*}
\xvec_{\lambda,\epsilon}^\star\in\arg\max_{\xvec\in\X} \left\{\xvec^\top\Umat \yvec_{\lambda,\epsilon}^\star -\frac{1}{\lambda}d_1^\epsilon(\xvec)\right\},\quad
\yvec_{\lambda,\epsilon}^\star\in\arg\min_{\yvec\in\Y} \left\{\xvec_{\lambda,\epsilon}^{\star,\top}\Umat \yvec +\frac{1}{\lambda}d_2^\epsilon(\yvec)\right\}.
\end{align*}
Then, by the fact that,the objectives of the problems above are strongly convex, one can immediately conclude that a pair of solutions $(\xvec_{\lambda,\epsilon}^\star, \yvec_{\lambda,\epsilon}^\star)$ must be unique and a $\C^\infty$.

We remark that equilibrium uniqueness in regularized-perturbed games plays a crucial in the construction underpinning our algorithm, as we will show in Section~\ref{sec:algo}.
Moreover, let us also notice that our approach to ensure equilibrium uniqueness is inspired by ideas taken from the \emph{quantal equilibrium}~\cite{mckelvey1995quantal}.
Indeed, it is easy to check that, by setting $\epsilon=0$ in $f_{\lambda,\epsilon}$, one immediately recovers the utility functions used to define quantal equilibria.

Next, we prove that the unique NE of a game $\G(\lambda,\epsilon)$ is indeed a $\delta$-approximate $\epsilon$-EFPE of the original EFG, where the approximation level $\delta$ linearly depends on the entity of the regularization component $1/\lambda$.
Formally:\footnote{All the proofs are in the Appendixes~\ref{app:1}~and~\ref{app:2}.}
\begin{restatable}[]{theorem}{deltaapproxefpe}\label{thm:connection_efpe}
	Given an EFG and $\lambda,\epsilon >0$, the unique NE $\zvec_{\lambda,\epsilon}^\star$ of $\G(\lambda,\epsilon)$ is a $O\left(  \frac{1}{\lambda} \right)$-approximate $\epsilon$-EFPE of the EFG.
\end{restatable}

\subsection{How to Select a Sequence Leading to an EFPE}

From Theorem~\ref{thm:connection_efpe}, one could na\"{\i}vely think that, by tracking a sequence of NEs of $\G(\lambda,\epsilon)$ as $\lambda\to+\infty$ and $\epsilon\to0$, it is possible to recover an EFPE of the original EFG.
In the following, we show that this is \emph{not} always the case, since, in order to identify the desired sequence of equilibria leading to an EFPE, one needs to carefully control how the parameters $\lambda$ and $\epsilon$ converge to $+\infty$ and zero, respectively. 

First, let us remark that, as an immediate corollary of Theorem~\ref{thm:connection_efpe}, we have that taking the limit as $\lambda\to+\infty$ before letting $\epsilon \to 0$ allows to recover an EFPE, since the theorem shows that a limit point as $\lambda\to+\infty$ of a sequence of $\zvec_{\lambda,\epsilon}^\star$ is an NE of a perturbed game.
However, doing so will result in losing all the benefits of the regularization component, which are needed in order to be able to design an efficient algorithm converging to an EFPE. 
Formally:
\begin{corollary}\label{cor:limits_efpe}
	Given any EFG, if $\zvec^\star \in \X \times \Y$ is such that
	$
		\lim_{\epsilon \to 0}\left[\lim_{\lambda\to\infty} \zvec_{\lambda,\epsilon}^\star\right] = \zvec^\star,
	$
	then $\zvec^\star $ is an EFPE.
\end{corollary}
\begin{figure}[t!]
	\centering
	\includegraphics[width=0.5\textwidth]{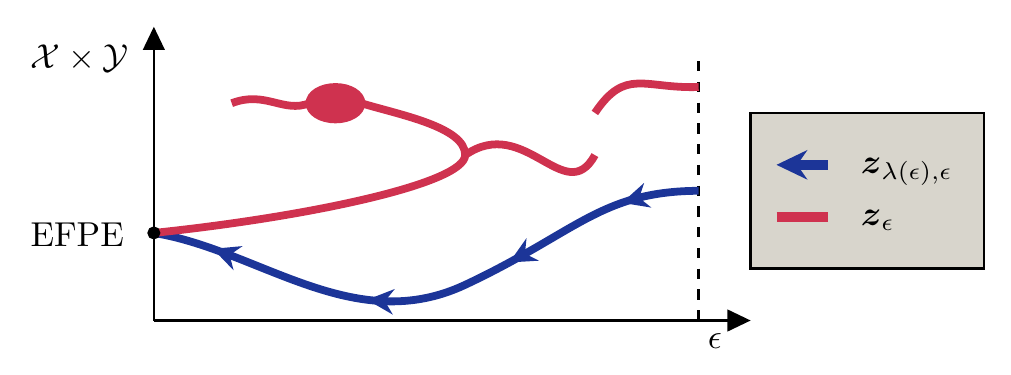}
	\caption{Examples of the ``smooth'' sequence of equilibria $\zvec_{\lambda(\epsilon),\epsilon}$ of regularized-perturbed games (in \emph{blue}) and the sequence of $\epsilon$-EFPEs $\zvec_{\epsilon}$ in the definition of EFPE (in \emph{red}).}
	\label{fig:homotopy}
\end{figure}
\begin{figure*}[!htp]
	\begin{minipage}[t]{0.42\textwidth}
		\begin{algorithm}[H]
			\small
			\caption{Compute EFPE}
			\label{alg:OOMD}
			\begin{algorithmic}[1]
				\Function{Compute-EFPE}{$\{ (\lambda_k,\epsilon_k) \}_{k \in \mathbb{N}} , \beta,\eta$}\Statex \hfill \textit{\textcolor{gray}{\begin{tabular}{l}
						$\{ (\lambda_k,\epsilon_k) \}_{k \in \mathbb{N}}$: As in the statement of Theorem~\ref{th:convergece}. \\
						$\beta > 1$: Phase growing rate.\\
						$\eta > 0$: Learning rate.
				\end{tabular}}}
				\State $\zvec^{(0)} \gets$ any $(\xvec,\yvec) \in \X \times \Y$; $k \gets 0$; $T\gets0$
				\While{\emph{not} exceeding time limit}
				\State $k \gets k+1$
				\State Instantiate game $\G_k \coloneqq \G(\lambda_k,\epsilon_k)$
				\State $T_k \gets \beta^k$; $T\gets T+T_k$ 
				\State $\zvec^{(k)}_{T_k} \gets$ \textsc{Solve}$(\G_k, T_k, \zvec^{(k-1)}_{T_{k-1}}, \eta)$
				\EndWhile
				\State \Return $\zvec^{(k)}_{T_k}$
				\EndFunction		%
			\end{algorithmic}%
		\end{algorithm}
	\end{minipage}
	\hfil
	\begin{minipage}[t]{0.58\textwidth}
		\begin{algorithm}[H]
			\small
			\caption{Solve game $\G_k$}
			\label{alg:OOMD_sub}
			\begin{algorithmic}[1]
				\Function{Solve}{$\G_k, T_k, \zvec^{(k)}_{-{1}/{2}},\eta$}
				\vspace{2pt}
				\For{$t= 0, \ldots, T_{k}-1$}
				\State $\xvec^{(k)}_{t+\frac{1}{2}} \gets \arg\max\limits_{\xvec\in\X}\left\{\xvec^\top  \Umat\yvec_t^{(k)}-\frac{1}{\lambda_k}{d_1^{\epsilon_k}(\xvec)}-\frac{1}{\eta}D_1^{\epsilon_k}\left(\xvec|\xvec^{(k)}_{t-\frac{1}{2}}\right)\right\}$
				\State $\xvec^{(k)}_{t+1}\gets \arg\max\limits_{\xvec\in\X}\left\{\xvec^\top \Umat\yvec_t^{(k)}-\frac{1}{\lambda_k}{d^{\epsilon_k}(\xvec)}-{\eta}D_1^{\epsilon_k}\left(\xvec|\xvec^{(k)}_{t+\frac{1}{2}}\right)\right\}$
				\State $\yvec^{(k)}_{t+\frac{1}{2}}\gets \arg\min\limits_{\yvec\in\Y}\left\{ \xvec^{(k),\top}_t\Umat\yvec+\frac{1}{\lambda_k}{d_2^{\epsilon_k}(\yvec)}+\frac{1}{\eta}D_2^{\epsilon_k}\left(\yvec|\yvec^{(k)}_{t-\frac{1}{2}}\right)\right\}$
				\State $\yvec^{(k)}_{t+1}\gets \arg\min\limits_{\yvec\in\Y}\left\{  \xvec^{(k),\top}_t\Umat\yvec+\frac{1}{\lambda_k}d_2^{\epsilon_k}(\yvec)+\frac{1}{\eta}D_2^{\epsilon_k}\left(\yvec|\yvec^{(k)}_{t+\frac{1}{2}}\right)\right\}$
				\EndFor
				\State \Return $\zvec^{(k)}_{{T_k}} \coloneqq \left(\xvec_{T_k}^{(k)},\yvec_{T_k}^{(k)}\right)$
				\vspace{2pt}
				\EndFunction		%
			\end{algorithmic}%
		\end{algorithm}
	\end{minipage}
\end{figure*}
Moreover, using \emph{any} sequence $\left\{ (\lambda_k,\epsilon_k) \right\}_{k \in \mathbb{N}}$ such that $\lambda_k \to +\infty$ and $\epsilon_k \to 0$ jointly as $k \to +\infty$ does \emph{not} always lead to EFPEs, as shown by the following proposition:
\begin{restatable}[]{proposition}{doublelimit}\label{th:doublelimit}
	There exist EFGs for which
	\[
		\lim\limits_{\lambda\to\infty}\left[\lim\limits_{\epsilon\to0} \zvec_{\lambda,\epsilon}^\star\right]\neq\lim\limits_{\epsilon\to0}\left[\lim\limits_{\lambda\to\infty} \zvec_{\lambda,\epsilon}^\star\right],
	\]
	and
	$
		\lim_{\lambda\to\infty}\left[\lim_{\epsilon\to0} \zvec_{\lambda,\epsilon}^\star\right]
	$
	is \emph{not} an EFPE.
\end{restatable}
Then, by well-known results on iterated limits~\citep{steinlage1971nearly}, we have that the double limit $\lim_{(\lambda, \epsilon)\to(+\infty,0)}\zvec^\star_{\lambda,\epsilon}$ may \emph{not} exist. 
This immediately implies that one cannot consider any arbitrary sequence $\left\{ (\lambda_k,\epsilon_k) \right\}_{k \in \mathbb{N}}$ in order to define a sequence of equilibria leading to an EFPE.
Indeed, the only guarantee that one gets is that the sequence leads to \emph{an} NE of the original EFG.
As a result, the sequence $\left\{ (\lambda_k,\epsilon_k) \right\}_{k \in \mathbb{N}}$ must be built in a specific way.

Next, we show that we need sequences $\left\{ (\lambda_k,\epsilon_k) \right\}_{k \in \mathbb{N}}$ that are defined so that the sequence made by the $\lambda_k$ converges faster than that of the $\epsilon_k$.
Formally:
\begin{restatable}[]{theorem}{convergence}\label{th:convergece}
	Given any EFG and a sequence $\left\{ (\lambda_k,\epsilon_k) \right\}_{k \in \mathbb{N}}$ such that $\lambda_k \to +\infty$, $\epsilon_k \to 0$ as $k \to +\infty$ and $1/\lambda_k<\epsilon_k^d$ for a sufficiently large $d \in \mathbb{N}$, $\lim_{k\to +\infty} \zvec^\star_{\lambda_k,\epsilon_k}$ is an EFPE. 
\end{restatable} 

The proof of Theorem~\ref{th:convergece} is based on the core idea that, by choosing $\lambda_k>1/\epsilon_k^d$, the resulting sequence $(1/\lambda_k,\epsilon_k)$ is ``close'' to the one defined by the iterated limit in Corollary~\ref{cor:limits_efpe}, which is guaranteed to be an EFPE.
Indeed, this amounts to showing that, for $k$ sufficiently large, the sequence $(1/\lambda_k,\epsilon_k)$ is close to $(0,\epsilon_k)$, where the latter sequence converges to an EFPE by Corollary~\ref{cor:limits_efpe}.

\subsection{The Need of Regularization}

Next, we argue why adding a regularization component is of paramount importance in our setting.
Indeed, one could argue that it would be more natural to track the sequence of $\epsilon$-EFPEs that appear in the definition of EFPE (Definition~\ref{def:eps_efpe}).
However, such a methodology would \emph{not} work.
Indeed, such a natural sequence, tough always admitting a limit point, can have discontinuities, bifurcations, multiple branches, and other irregularities, which are drawback inherited by the convoluted structure of NE in linear games.
In particular, non-uniqueness alone would doom any method that relies on last-iterate convergence, since uniqueness is a key assumption in such methods.

On the other hand, our sequence of equilibria defined in terms of regularized-perturbed games is $\C^\infty$-differentiable and it identifies a unique branch, while sharing the same limit point with the sequence in the definition of EFPE.
This renders our sequence of equilibria much easier to track.

Figure~\ref{fig:homotopy} provides a depiction of the sequence used in the definition of EFPE and the one generated by our regularized-perturbed games.

\section{Efficient Last-iterate Algorithm for EFPEs}\label{sec:algo}

We are now ready to introduce our \emph{last-iterate} algorithm that converges to an (exact) EFPE (Algorithm~\ref{alg:OOMD}).
The algorithm works by tracking a sequence of equilibria $\zvec^\star_k \coloneqq \zvec^\star_{\lambda_k, \epsilon_k}$ of regularized-perturbed games $\G_k\coloneqq\G(\lambda_k,\epsilon_k)$ by letting $k \to +\infty$, where the latter games are defined by means of a sequence $\left\{ (\lambda_k,\epsilon_k) \right\}_{k \in \mathbb{N}}$ as in the statement of Theorem~\ref{th:convergece}.

Algorithm~\ref{alg:OOMD} works in phases.
For every $k \in \mathbb{N}$, the $k$-th phase of the algorithm is devoted to finding a suitable approximation of the equilibrium $\zvec^\star_k$ of $\G_k$.
This is achieved by performing $T_k \coloneqq \beta^k$ (for a given $\beta > 1$) iterations of a last-iterate sub-procedure (Algorithm~\ref{alg:OOMD_sub}) that converges linearly to an approximate equilibrium of $\G_k$ (see Theorem~\ref{thm:algo_conv}).
We remark that, in any given phase $k$, finding an approximation of $\zvec^\star_k$ is sufficient, since the sequence made by the $\zvec_k^\star$ has the only purpose of identifying EFPEs as its limit points.

Algorithm~\ref{alg:OOMD_sub} is an extension of the OOMD algorithm, when instantiated for sequence-form strategy sets.
Similarly to the OOMD algorithm, Algorithm~\ref{alg:OOMD_sub} performs suitably-defined intermediate updates, labeled with the subscripts $t + \frac{1}{2}$. 
This enables the algorithm to converge last iterate.

Notice that, in our setting, it is \emph{not} possible to use the standard update rule of OOMD, which would work by directly feeding the algorithm with the gradients of the utility functions $f_{\lambda_k,\epsilon_k}$ of regularized-perturbed games $\G_k$.
This is because such functions are non-smooth due to the $d_i^{\epsilon_k}$ having unbounded gradient norm near the boundaries of $\X^{\epsilon_k}\times\Y^{\epsilon_k}$.
Indeed, $d_i^{\epsilon_k}$ are Legendre functions~\citep{cesa2006prediction}.
In order to circumvent this issue, we exploit an idea first introduced by the \emph{composite objective mirror descent} algorithm~\citep{duchi2010composite}, which consists in avoiding the linearization of the non-smooth part in the utility function in order to get away with the non-Lipschitzness of the gradients.
The following theorem formally states the last-iterate convergence guarantees of Algorithm~\ref{alg:OOMD_sub}.

\begin{restatable}[]{theorem}{linearconv}\label{thm:algo_conv}
	Given $\eta\le {1}/{\sqrt{2}{\|\Umat\|_2}}$, for every phase $k \in \mathbb{N}$ and round $t\in\{1,\ldots,T_k\}$, Algorithm~\ref{alg:OOMD_sub} guarantees that:
	\[
	\|\zvec_k^\star-\zvec^{(k)}_{t}\|_2\le2\left[2D^{\epsilon_k}\left( \zvec_k^\star|\zvec^{(k)}_0 \right)\right]^{1/2}\left(\frac{\lambda_k}{\lambda_k+\eta}\right)^{t/2},
	\]
	where $D^{\epsilon_k}(\zvec|\tilde\zvec):=D^{\epsilon_k}_1(\xvec|\tilde\xvec)+D^{\epsilon_k}_2(\yvec|\tilde\yvec)$.
\end{restatable}

We also remark that the technique we employ in the proof of the theorem above can be generalized to achieve last-iterate convergence to equilibria of EFGs for \emph{any} dilated Legendre regularization functions $d_i$.
Indeed, our results do \emph{not} rely on the specific shape of the $d_i$, and the ones we used are specifically tailored only to get convergence to EFPEs.

\subsection{Convergence Analysis}

By exploiting the convergence analysis of Algorithm~\ref{alg:OOMD_sub} (Theorem~\ref{thm:algo_conv}), the following theorem formally proves the convergence guarantees of Algorithm~\ref{alg:OOMD} to (exact) EFPEs.
\begin{restatable}[]{theorem}{anytimeconvtwo}\label{th:anytime_conv2}
	Given any sequence $\{(\lambda_k,\epsilon_k)\}_{k\in\mathbb{N}}$ defined as in Theorem~\ref{th:convergece} satisfying $\eta\le\lambda^2_k\le\eta\beta^k/2$ for every $k \in \mathbb{N}$, Algorithm~\ref{alg:OOMD} grantees that $\lim_{k\to\infty}\zvec_{T_k}^{(k)}$ is an EFPE.
\end{restatable}
Moreover, we can exploit the strongly convex-concave structure of the utility functions $f_{\lambda,\epsilon}$ to prove Nash gap guarantees for regularized-perturbed games $\G_{\lambda,\epsilon}$, which can then be combined with Theorem~\ref{thm:connection_efpe} to obtain guarantees on the exploitability of $\zvec_{T_k}^{(k)}$ in the original EFG.
Formally:
\begin{restatable}[]{theorem}{saddelpointgap}\label{thm:saddelpointgap}
	Let $\tilde\zvec \coloneqq(\tilde\xvec,\tilde\yvec)\in\X\times\Y$ be such that it holds $\|\zvec^\star_{\lambda,\epsilon}-\tilde\zvec\|_2\le\nu$ for some $\nu > 0$, then we have:
	\[
	Gap(\tilde\xvec, \tilde\yvec)\le\nu |\Sigma|^{3/2}+{8C}/{\lambda},
	\]
	where $|\Sigma|\coloneqq\max\{|\Sigma_1|,|\Sigma_2|\}$ and $C > 0$ is a constant that depends polynomially in $\max_{i\in\{1,2\}}|\I_i|$, $\max_{I\in\I_1\cup\I_2}\alpha_I$, and $\max_{I\in\I_1\cup\I_2}\log n_I$.
\end{restatable}
As a direct corollary, we can prove the following guarantees in terms of Nash gap for Algorithm~\ref{alg:OOMD}.
\begin{restatable}[]{corollary}{matrixgapfinal}\label{cor:matrixgapfinal}
	Given any sequence $\{(\lambda_k,\epsilon_k)\}_{k\in\mathbb{N}}$ as in Theorem~\ref{th:anytime_conv2}, at the end of every phase $k \in \mathbb{N}$ of Algorithm~\ref{alg:OOMD}:
	\[
		Gap\left( \xvec^{(k)}_{T_k},\yvec^{(k)}_{T_k} \right)\le\frac{4}{\lambda_k}|\Sigma|^{\frac{3}{2}}\left[2D^{\epsilon_k}(\zvec_{k}^\star|\zvec_0^{(k)})\right]^{\frac{1}{2}} +\frac{8C}{\lambda_k},
	\]
	where $\Sigma$ and $C$ are as in Theorem~\ref{th:anytime_conv2}.
\end{restatable}

The result above justifies the need of the exponential growth of the number of rounds $T_k$.
Indeed, it shows that, after $T=\sum_{j=1}^k\beta^j=O(\beta^k)$ rounds, one gets a Nash gap of the order of $O(1/\lambda_k)$.
Thus, by choosing $\lambda_k=O(\beta^{k/2})$ (which is the largest order of $\lambda_k$ allowed by Corollary~\ref{cor:matrixgapfinal}), we have a bound of $O(1/\sqrt{T})$ on the Nash gap.

Notice that Theorem~\ref{th:convergece} shows asymptotic convergence in terms of $\ell_2$-distance to an EFPE, while Corollary~\ref{cor:matrixgapfinal} shows convergence in terms of Nash gap.
Determining whether it is possible to show finite convergence rates also for the $\ell_2$-distance to the set of EFPEs requires additional research.
The difficulty of such an inquiry is that, even in (non-perturbed) normal-form games, our problem reduces to the one of finding convergence rates to limit-quantal equilibria, which is still an open problem.

\subsection{Efficient Decomposition}

The pseudo-code formulation of Algorithm~\ref{alg:OOMD_sub} does \emph{not} allow it to scale efficiently, since at each iteration it requires the solution of four convex optimization problems defined over the whole sequence-form strategy set $\X$ or $\Y$.

Next, we show that each iteration of Algorithm~\ref{alg:OOMD_sub} can indeed be implemented \emph{efficiently} by directly working on the game tree.
In particular, we show how to implement each update in Algorithm~\ref{alg:OOMD_sub} so that: (i) it is performed recursively by means of a bottom-up visit of one player's infosets; and (ii) at every visited infoset, it only requires the application of a \emph{local} update that is done by applying a closed-form formula.

In order to do that, we exploit a technique that has been originally introduced in~\citep{hoda2010smoothing, farina2021better, lee2021last}.
By focusing on the first player (analogous facts hold for the second player), we have that if the following two facts hold:
\begin{enumerate}
	\item the update rule can be expressed as a \emph{proximal gradient update}, \emph{i.e.}, as the problem of computing the conjugate gradient $\nabla d_1^{\epsilon,*}(\tilde\gvec) \coloneqq \arg\max_{\xvec\in\X}\left\{\xvec^\top\tilde\gvec-d_1^\epsilon(\xvec)\right\}$ of the dilated function $d_1^{\epsilon}$ for some $\tilde{\gvec} \in \mathbb{R}^{|\Sigma_1|}$, and
	\item for all $I\in\I_1$ and $\tilde\gvec\in\mathbb{R}^{n_I}$, the \emph{local conjugate gradient} $\nabla d_{\Delta_I}^{\epsilon,*}(\tilde\gvec)\coloneqq \arg\max_{\wvec\in\Delta_{I}}\left\{\wvec^\top\tilde\gvec-d_1^\epsilon(\wvec)\right\}$ and the \emph{local gradient} $\nabla d_{\Delta_I}^\epsilon(\tilde\gvec)$ have a closed-form solution,
\end{enumerate}
then the overall updated can be computed efficiently in terms of requirements (i) and (ii) described above.

However, Algorithm~\ref{alg:OOMD_sub} employs update rules that, a priori, are differently to those studied in~\citep{hoda2010smoothing} and its follow ups.
Thankfully, the following theorem show that the approach described above can still be employed on the update rules of Algorithm~\ref{alg:OOMD_sub}.
Formally:
\begin{restatable}[]{theorem}{efficencyone}\label{thm:efficency_1}
	\textbf{(i)} The updates of Algorithm~\ref{alg:OOMD_sub} of the form
	\begin{align*}
		\arg\max\limits_{\xvec\in\X}\left\{\xvec^\top\gvec-\frac{d_1^\epsilon(\xvec)}{\lambda}-\frac{1}{\eta}D_1^\epsilon(\xvec|\tilde{\xvec})\right\}
	\end{align*}
	for some vectors $\gvec, \tilde{\xvec}\in \mathbb{R}^{|\Sigma_1|}$ and $\epsilon >0$ can be formulated as the computation of the conjugate gradient $\nabla d_1^{\epsilon,*}(\tilde\gvec)$ for a suitably-defined, efficiently-computable vector $\tilde\gvec \in \mathbb{R}^{|\Sigma_1|}$.

	\textbf{(ii)} For every $\epsilon > 0$, $I \in \I_1$, and $\tilde\gvec\in\mathbb{R}^{n_I}$, it holds:
	\[
	\nabla d_{\Delta_I}^{\epsilon,*}(\tilde\gvec)[a]=(1-\epsilon \, n_I)\frac{e^{\tilde\gvec[a]}}{\sum_{b\in A(I)} e^{\tilde\gvec[b]}}+\epsilon \quad \forall a \in A(I).
	\]
	Moreover, the local gradient $\nabla d_{\Delta_I}^\epsilon(\tilde\gvec)$ can be computed as $\nabla d_{\Delta_I}^{\epsilon}(\tilde\gvec)[a]=1+\log(\tilde\gvec[a]-\epsilon)$ for all $a\in A(I)$.

\end{restatable}
\section{Experimental Evaluation}\label{sec:exper}

We conclude the paper by experimentally evaluating our algorithm (Algorithm~\ref{alg:OOMD}) on a standard testbed of EFGs.
In particular, we consider two simplified versions of Poker, called \emph{Kuhn}~\citep{kuhn1950simplified} and \emph{Leduc}~\citep{southey2005bayes}, and a two-player card game called \emph{Goofspiel}~\cite{ross1971goofspiel}.\footnote{See Appendix~\ref{app:3} for a description of the games.}

We evaluate our algorithm in terms of two metrics. 
The first one is the \emph{Nash gap}, which is the standard metric employed to evaluate NE approximation.
The second metric is the \emph{average infoset regret}, called $R_I$ for short.
Given a pair of players' strategies, this is defined as the average of the regrets that such strategies incur at all the infosets of the game.
Specifically, the regret at an infoset is computed by assuming that such infoset is reached with probability one and by applying the Bayes' rule to get a distribution encoding the probability of reaching each node in the infoset.
Then, the regret is defined as the increase in utility from best-responding at that infoset and at all the subsequent ones.

We compare Algorithm~\ref{alg:OOMD} with two baseline algorithms for computing NEs, namely the CFR algorithm~\citep{zinkevich2007regret} and the OOMD algorithm instantiated for EFG strategy spaces, as in~\citet{farina2019optimistic}.
We run Algorithm~\ref{alg:OOMD} with $\beta=1.001$, $\eta=2$, $\epsilon_k=0.9999^k$, and $\lambda_k=1/\epsilon_k^2$.
Moreover, we also compare the performances of Algorithm~\ref{alg:OOMD} with those obtained by the same algorithm while keeping the value of $\epsilon$ unchanged across its execution and eliminating regularization (\emph{i.e.}, $1/\lambda=0$).
%
In particular, we test $\epsilon = 0.01$ and $\epsilon = 0.001$.
Since these baselines can also be viewed as the OOMD algorithm obtained by dilation of regularizer $d_{\Delta_I}^\epsilon$, we name them OOMD($\epsilon$).

Figure~\ref{fig:exp} shows some results of our experimental evaluation (additional ones are in Appendix~\ref{app:3}).\footnote{For Algorithm~\ref{alg:OOMD}, OOMD and OOMD$(\epsilon)$ we considered the last-iterates, while for CFR we considered the time average.}
As it is clear from the plots, Algorithm~\ref{alg:OOMD} outperforms of orders of magnitude the others in terms of average infoset regret.
This was expected, since our algorithm is specifically tailored for converging to an EFPE, while the others only guarantee convergence to an NE or only find $\epsilon$-EFPEs.
Notice that Algorithm~\ref{alg:OOMD} also out-competes its variations that keep the value of $\epsilon$ unchanged, showing the importance of tuning regularization and perturbation terms jointly in order to converge to exact EFPEs.
Indeed, every method that finds approximate equilibria of the perturbed game provably fails in getting both zero Nash gap and zero average infoset regret.

In terms of Nash gap, Algorithm~\ref{alg:OOMD} beats the baselines in \emph{Kuhn}, it is (almost) matched by the CFR algorithm in \emph{Leduc}, while it is outperformed by both the CFR and OOMD algorithms in \emph{Goofspiel}.
This is \emph{not} surprising, since such algorithms are designed with the only (easier) objective of converging to an NE, and, thus, they perform better than Algorithm~\ref{alg:OOMD} in doing so on certain games.
Indeed, the fact that Algorithm~\ref{alg:OOMD} is able to compete or even beat both the CFR and the OOMD algorithms on some game instances (see \emph{Kuhn} and \emph{Leduc}) was unexpected.
This shows that the techniques that we employ not only do converge last-iterate to EFPEs but in also have the potential of leading to the design of algorithms that are superior to current state-of-the-art equilibrium-computation algorithms.

\begin{figure}[t]
	\centering
	\begin{subfigure}{.49\columnwidth}
		\centering
		\includegraphics[width=\linewidth]{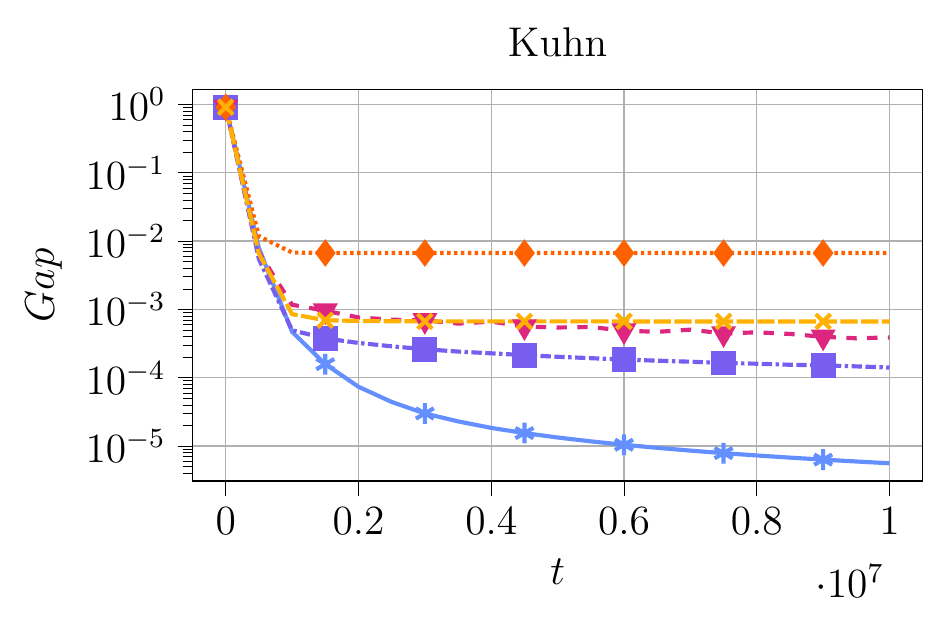}
	\end{subfigure}%
	\hfill
	\begin{subfigure}{.49\columnwidth}
		\centering
		\includegraphics[width=\linewidth]{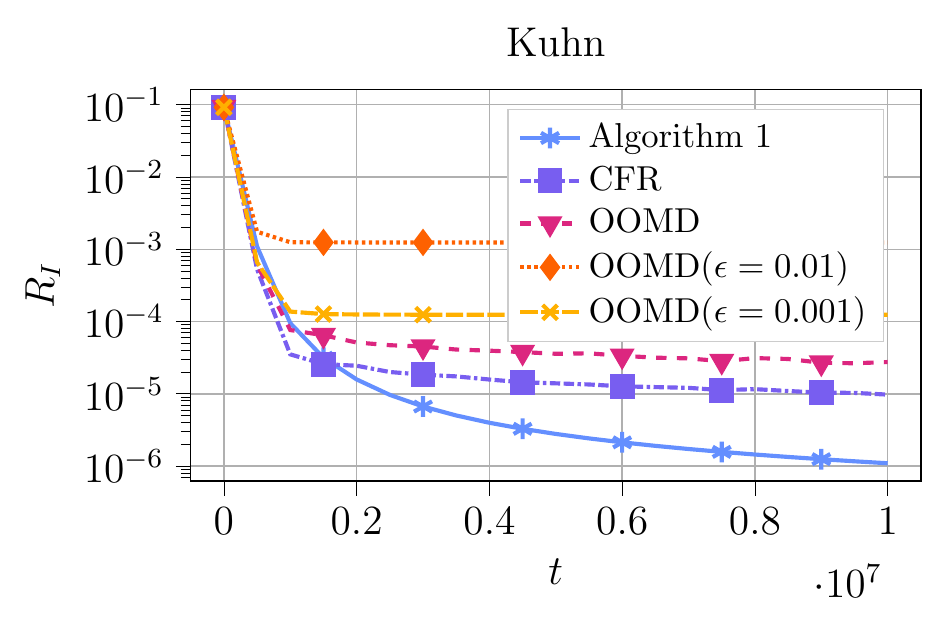}
	\end{subfigure}

	\begin{subfigure}{.49\columnwidth}
		\centering
		\includegraphics[width=\linewidth]{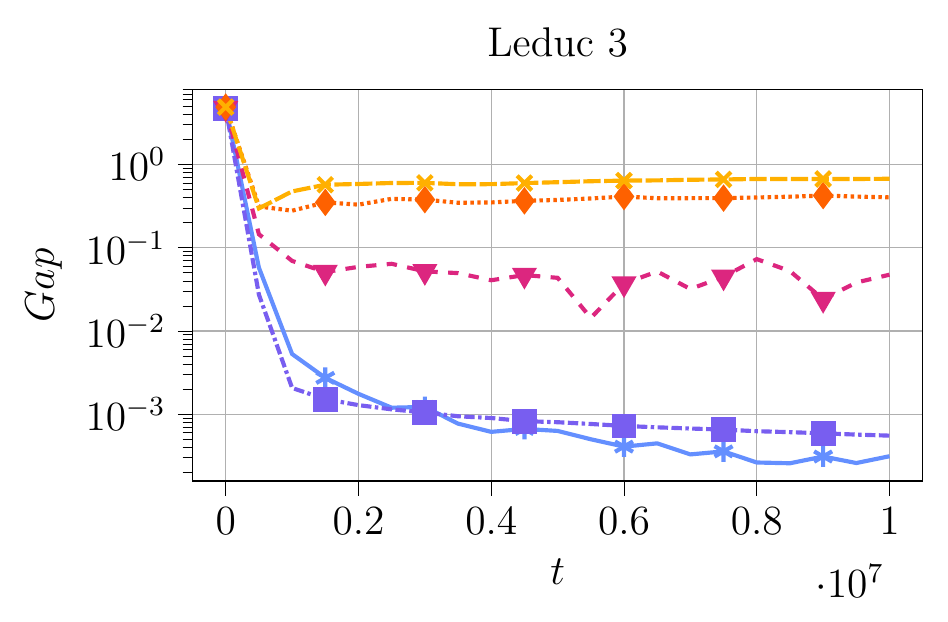}
	\end{subfigure}%
	\hfill
	\begin{subfigure}{.49\columnwidth}
		\centering
		\includegraphics[width=\linewidth]{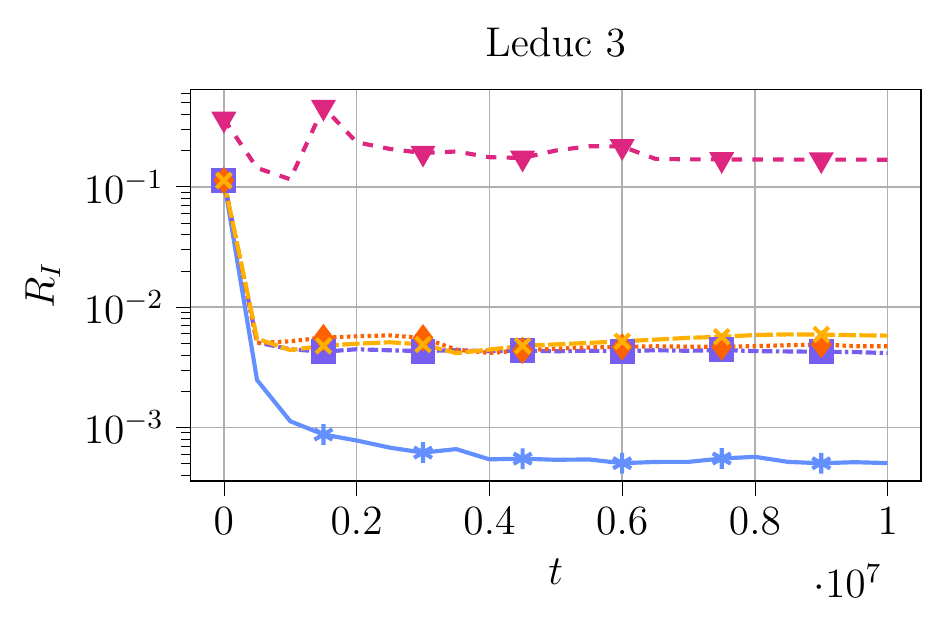}
	\end{subfigure}

	\begin{subfigure}{.49\columnwidth}
		\centering
		\includegraphics[width=\linewidth]{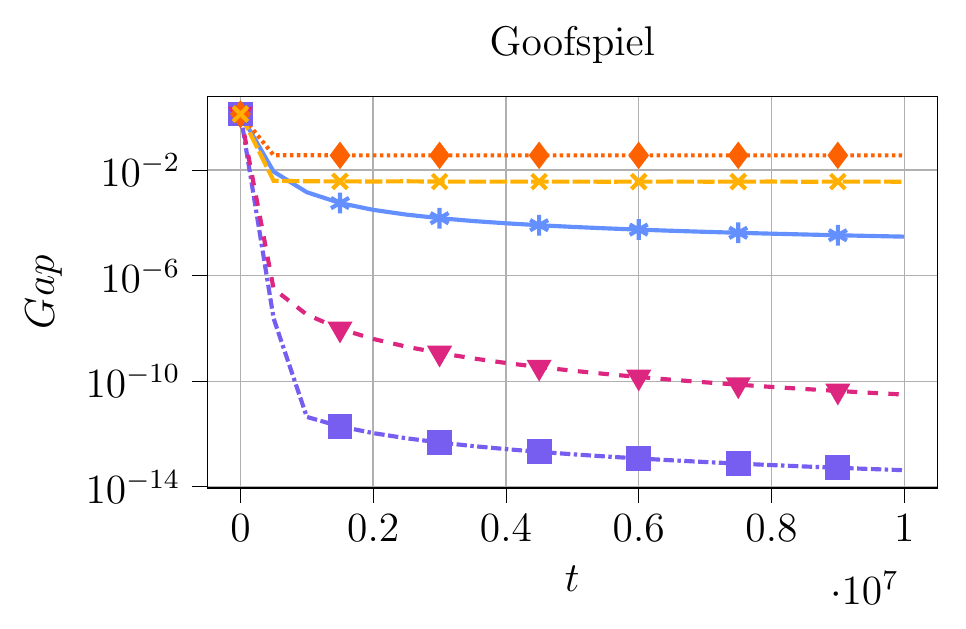}
	\end{subfigure}%
	\hfill
	\begin{subfigure}{.49\columnwidth}
		\centering
		\includegraphics[width=\linewidth]{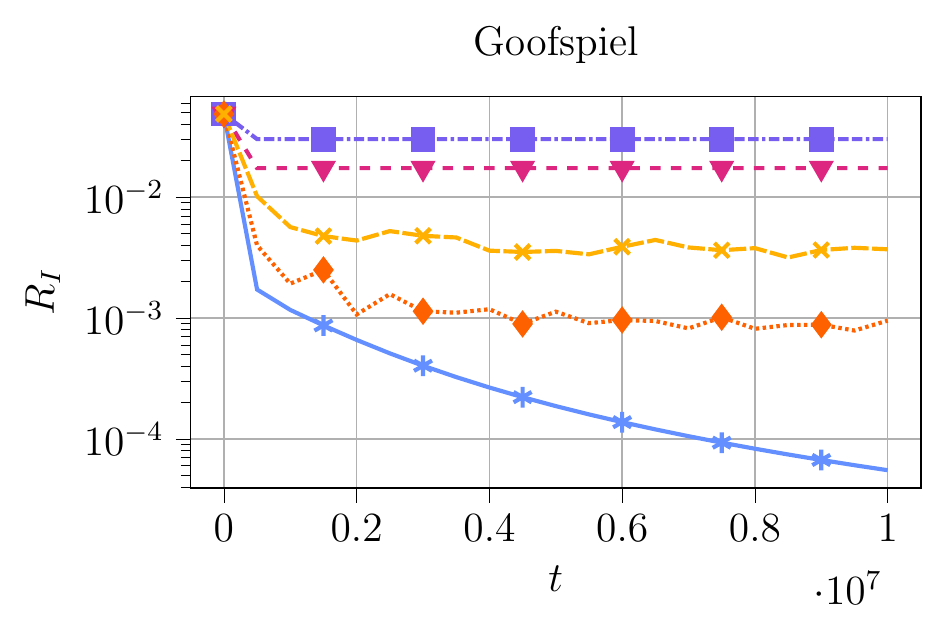}
	\end{subfigure}
	\caption{Results of the experimental evaluation. Algorithm~\ref{alg:OOMD} is compared with the baselines in terms of Nash gap (\emph{Left}) and average infoset regret (\emph{Right}).
}
	\label{fig:exp}
\end{figure}

\clearpage

\bibliographystyle{named}
\bibliography{bibliography}

\onecolumn
\appendix
\clearpage

\section*{Appendix Index}

The appendix is structured as follows:

\begin{itemize}
	\item Appendix~\ref{app:1} presents the proofs omitted from Section~\ref{sec:methods}.
	\item Appendix~\ref{app:2} provides the proofs and additional details omitted from Section~\ref{sec:algo} in the construction of Algorithm~\ref{alg:OOMD}.
	\item Appendix~\ref{app:3} provides additional experiments details and results.
\end{itemize}

\section{Proofs omitted from Section~\ref{sec:methods}}\label{app:1}

\begin{lemma}\label{lem:strong_convexity2}
	For every $I \in \I_1 \cup \I_2$ and $\epsilon >0$, the function $d_{\Delta_I}^\epsilon$ is $1$-strongly convex with respect to the Euclidean norm.
\end{lemma}

\begin{proof}
	For every $ \wvec\in \Delta_{I}$, we have that:
	\[
	\frac{\partial^2 d_{\Delta_I}^\epsilon(\wvec)}{\partial\wvec[i]\partial \wvec[j]}=
	\begin{cases}
		\frac{1}{\wvec[i]-\epsilon}&\text{if}\quad i=j\\
		0&\text{otherwise}
	\end{cases}.
	\]
	Thus, $\nabla^2 d_{\Delta_I}^\epsilon(\wvec)\ge \,\boldsymbol{I}_{n_I}$, where $\boldsymbol{I}_n$ denotes the $n$-dimensional identity matrix. This completes the proof.
\end{proof}

\begin{lemma}\label{lem:strongconvexity}
	For any $\epsilon > 0$, the functions $d_1^{\epsilon}$ and $d_2^{\epsilon}$ are $1$-strongly convex with respect to the Euclidean norm.
\end{lemma}

\begin{proof}
	This follows from Lemma~\ref{lem:strong_convexity2}, the expression of the weights $\alpha_I$ in $d_{\Delta_I}^\epsilon$ in the definitions of $d_1^\epsilon(\cdot)$ and $d_2^\epsilon(\cdot)$, and by using~\citet[Corollary~1]{farina2019optimistic}.
\end{proof}

\begin{lemma}\label{lem:bounded_d}
	We have that $|d_1^{\epsilon}(\xvec)|\le C$ and $|d_2^{\epsilon}(\yvec)|\le C$ for all $\xvec,\yvec\in\X\times\Y$ and any $\epsilon\le \min\limits_{I\in\I_1\cup\I_2}\frac{1}{2n_I}$.
	Moreover $C\le\|\boldsymbol{\alpha}\|_\infty\cdot \max\limits_{i\in\{1,2\}}|I_i|\max\limits_{I\in \I_1\cup\I_2}\log(2n_I)$, where $\boldsymbol{\alpha}\in\mathbb{R}^{|\I_1|+|\I_2|}$ is a vector that contains all the components $\alpha_I$ with $I\in\I_1\cup\I_2$.
\end{lemma}

\begin{proof}
	Consider the following inequalities: 
	\begin{align*}
	\max\limits_{\xvec\in\X} |d_1^\epsilon(\xvec)|
	&\le\sum\limits_{I\in\I_1}{\alpha_I}\max\limits_{\wvec\in\Delta_{I}}|d_{\Delta_I}^\epsilon(\wvec)|\\
	&\le \sum\limits_{I\in\I_1}\alpha_I\log\left(\frac{n_I}{1-\epsilon\cdot n_I}\right)\\
	&\le\sum\limits_{I\in\I_1}\alpha_I\log(2n_I):=C_\X,
\end{align*}
where we used that the maximum of $d_{\Delta_I}^\epsilon(\wvec)$ is attained in the center $\wvec[a]=1/n_I$, for all $a\in A(I)$, and the last inequality follows from $\epsilon\le \min\limits_{I\in\I_1\cup\I_2}\frac{1}{2n_I}$.
We can define $C_\Y$ analogously for the second player and take $C:=\max\{C_\X,C_\Y\}$.

Then define the vector $\boldsymbol{\alpha}$ as the vector that contains the components $\alpha_I$ for the first player \emph{and} second player.
Then, by Holder inequality, we have that $C\le\|\boldsymbol{\alpha}\|_\infty\cdot \max\limits_{i\in\{1,2\}}|I_i|\max\limits_{I\in \I_1\cup\I_2}\log(2n_I)$.

\end{proof}

\begin{lemma}\label{lem:maxinequality}
	The following inequalities holds:
	\begin{align}
		\left|\max\limits_{\xvec}f(\xvec)-\max\limits_{\yvec}g(\yvec)\right|\le\max\limits_{\xvec}|f(\xvec)-g(\xvec)|\label{eq:lemma11}\\
		\left|\min\limits_{\xvec}f(\xvec)-\min\limits_{\yvec}g(\yvec)\right|\le\max\limits_{\xvec}|f(\xvec)-g(\xvec)|\label{eq:lemma12}
	\end{align}
\end{lemma}
\begin{proof}
	Consider the following inequality:
	\[
	f(\xvec)\le |f(\xvec)-g(\xvec)|+g(\xvec)
	\]
	Applying the $\max$ operator to both sides of the previous equation and observing that $\max\limits_{\xvec}(|f(\xvec)-g(\xvec)|+g(\xvec))\le\max\limits_{\xvec}|f(\xvec)-g(\xvec)|+\max\limits_{\xvec}g(\xvec))$ we can rearrange the inequality to obtain:
	\[
	\max\limits_{\xvec}f(\xvec)-\max\limits_{\xvec}g(\xvec)\le \max\limits_{\xvec}|f(\xvec)-g(\xvec)|.
	\]
	We can in the same way obtain that:
	\[
	\max\limits_{\xvec}g(\xvec)-\max\limits_{\xvec}f(\xvec)\le \max\limits_{\xvec}|f(\xvec)-g(\xvec)|,
	\]
	which combined with the above let us conclude Equation~\eqref{eq:lemma11}. Now Equation~\eqref{eq:lemma12} follows from considering  Equation~\eqref{eq:lemma11} with $-f(\xvec)$ instead of $f(\xvec)$ and $-g(\xvec)$ instead of $g(\xvec)$.
\end{proof}

\deltaapproxefpe*

\begin{proof}
	Let us consider the function $\tilde{f}_\epsilon: \X \times \Y \to \mathbb{R}$ such that, for every $\xvec \in \X$ and $\yvec \in \Y$, it holds:
	\[
	\tilde{f}_\epsilon(\xvec, \yvec) \coloneqq \xvec^\top \Umat \yvec - \mathbb{I}_{\X^{\epsilon}}(\xvec)+ \mathbb{I}_{\Y^{\epsilon}}(\yvec),
	\]
	where $\mathbb{I}_{\mathcal{W}}(\wvec)$ is $0$ if $\wvec\in \mathcal{W}$, while it is $+\infty$ if $\wvec\not\in \mathcal{W}$.

	First, it is easy to check that any $(\xvec,\yvec) \in \X \times \Y$ such that $\tilde{f}_\epsilon(\xvec, \yvec) = \max_{\tilde\xvec \in \X}\min_{\tilde\yvec \in \Y} \tilde{f}_\epsilon(\tilde\xvec, \tilde\yvec)$ is an $\epsilon$-EFPE of the original EFG (see Definition~\ref{def:eps_efpe}).

	Moreover, for any $(\xvec, \yvec)\in \X^{\epsilon}\times\Y^{\epsilon}$ it is easy to show that $|f_{\lambda,\epsilon}-\tilde f_{\epsilon}|=O(1/\lambda)$:
	\begin{align*}
		|f_{\lambda,\epsilon}(\xvec,\yvec)-\tilde{f}_\epsilon(\xvec,\yvec)|&\le\frac{\left|d_1^\epsilon(\xvec)-d_2^\epsilon(\yvec)\right|}{\lambda}\le \frac{2C}{\lambda},
	\end{align*}
	where $C$ is defined in Lemma~\ref{lem:bounded_d}.
	
	Finally, let us consider the unique NE $\zvec^\star_{\lambda,\epsilon}=(\xvec_{\lambda,\epsilon}^\star,\yvec_{\lambda,\epsilon}^\star)$ of $\G(\lambda,\epsilon)$.
	By Lemma~\ref{lem:maxinequality}, we have that:
	\begin{align*}
		\max_{\xvec\in  \X^\epsilon} \tilde{f}_\epsilon(\xvec,\yvec^\star_{\lambda,\epsilon})-  \min_{\yvec\in \Y^\epsilon} \tilde{f}_\epsilon(\xvec^\star_{\lambda,\epsilon},\yvec)
		& \leq \left|\max_{\xvec \in \X^\epsilon}\tilde{f}_\epsilon(\xvec,\yvec^\star_{\lambda,\epsilon})-\max_{\xvec \in  \X^\epsilon} f_{\lambda,\epsilon}(\xvec,\yvec^\star_{\lambda,\epsilon})\right|+
		\left|\min_{\yvec \in \Y^\epsilon}\tilde{f}_\epsilon(\xvec^\star_{\lambda,\epsilon},\yvec)-\min_{\yvec \in \Y^\epsilon} f_{\lambda,\epsilon}(\xvec^\star_{\lambda,\epsilon},\yvec)\right|\\
		& \le \frac{4C}{\lambda},
	\end{align*}
	which concludes the proof.
\end{proof}

\doublelimit*

\begin{proof}

	Let us consider a game in which the two players play simultaneously and only once, having each of them three different actions available (\emph{i.e.}, a $3 \times 3$ game in normal form).
	The following matrix encodes the first player's payoffs for all the possible combinations of players' actions: 
	\[
	\left[ {\begin{array}{ccc}
			0.3 & 0.5 & 0.3 \\
			0.7& 0.3 & 0.7 \\
			0.6 & 0.2 & 0.2
	\end{array} } \right]
	\]
	
	By using the fact that any quantal equilibrium enjoys the "independence
	of irrelevant alternatives" property~\cite{mckelvey1995quantal}, we can prove that it holds:
	\[
	\lim\limits_{\lambda\to\infty}\left[\lim\limits_{\epsilon\to0} \zvec_{\lambda,\epsilon}^\star\right]=\left(\left[ {\begin{array}{ccc}
			1/2\\
			1/2 \\
			0
	\end{array} } \right],\left[ {\begin{array}{ccc}
			1/6\\
			2/3 \\
			1/6
	\end{array} } \right]\right),
	\]
	where we expressed player's strategies as the probability distributions that they induce over the three actions for ease of presentation.
	Moreover, the only EFPE of the game is:
	\[\zvec^\star=
	\left(\left[ {\begin{array}{ccc}
			1/2\\
			1/2 \\
			0
	\end{array} } \right],\left[ {\begin{array}{ccc}
			0\\
			2/3 \\
			1/3
	\end{array} } \right]\right),
	\]
	since this is the unique NE that eliminates weakly dominated strategies, and, thus, it is an~\citep[{Corollary~2.2.6}]{van1991stability}.
	This proves the proposition.
\end{proof}

\convergence*

\begin{proof}
	By Corollary~\ref{cor:limits_efpe}, we have that
	\[
	\zvec^\star\coloneqq\lim\limits_{\epsilon\to0}\left[\lim\limits_{\lambda\to\infty} \zvec_{\lambda,\epsilon}^\star\right]
	\]
	is an EFPE.
	Let $\zvec_\epsilon^\star \coloneqq \lim_{\lambda\to\infty}\zvec^\star_{\lambda,\epsilon}$ for all $\epsilon > 0$.
	By definition of limit, for every $\tau >0$ there exists $R_{\epsilon}(\tau) \in \mathbb{R}_+$ such that $|\zvec^\star_{\lambda,\epsilon}-\zvec^\star_{\epsilon}|\le \tau$ for all $\lambda \in \mathbb{R}: \lambda>R_{\epsilon}(\tau)$.
	Moreover, by looking at the outer limit in the definition of $\zvec^\star$ we have that, for every $\tau^\prime>0$, there exists $H(\tau^\prime) \in \mathbb{R}_+$ such that $|\zvec^\star_\epsilon-\zvec^\star|\le\tau^\prime$ for all $\epsilon \in \mathbb{R} : |\epsilon|\le H(\tau^\prime)$.
	
	By using the triangular inequality, for every $\tau >0$:
	\begin{align*}\label{eq:triangular}
		|\zvec^\star_{\epsilon,\lambda}-\zvec^\star|\le |\zvec^\star_{\epsilon,\lambda}-\zvec^\star_\epsilon|+|\zvec^\star_\epsilon-\zvec^\star|\le \tau
	\end{align*}
	for all $\lambda \in \mathbb{R}: \lambda>R_{\epsilon}(\tau/2)$ and $\epsilon \in \mathbb{R} : |\epsilon|\le H(\tau/2)$.

	Moreover, if $\epsilon_k \leq H(\tau/2)$, then $\lambda_k >R_{\epsilon_k}(\tau/2)$.
	This follows from the following inequalities:
	\[
	\lambda_k > \frac{1}{\epsilon_k^d} \geq \frac{1}{H(\tau/2)^d} \geq R_{\epsilon_k}(\tau/2),
	\]
	where the last inequality holds for any sufficiency large $d$ and the fact that $R_{\epsilon_k}(\tau/2)$ can be upper bounded for every $k$ (otherwise $\zvec_0^\star \coloneqq \lim_{\lambda\to\infty}\zvec^\star_{\lambda,0}$ would not exists).
	
	As a result, for every $\tau > 0$ and for all $k \in \mathbb{N}$ such that $\epsilon_k \leq H(\tau/2)$, it holds:
	\begin{align*}
		|\zvec^\star_{\epsilon_k,\lambda_k}-\zvec^\star|\le \tau,
	\end{align*}
	which concludes the proof.
\end{proof}

\section{Proofs omitted from Section~\ref{sec:algo}}\label{app:2}

In this section is convenient to consider the join updates for the first and second player. If we define $\zvec=(\xvec,\yvec)\in\X\times\Y$, then it is easy to verify that updates of the form:
\begin{align}
	\hat\xvec=\arg\max\limits_{\xvec\in\X}\left\{ \xvec^\top \Umat\yvec_0-\frac{1}{\lambda}d_1^\epsilon(\xvec)-\frac{1}{\eta}D_1^{\epsilon}(\xvec|\tilde{\xvec})\right\}\\
	\hat\yvec=\arg\min\limits_{\yvec\in\Y}\left\{ \xvec_0^\top \Umat\yvec+\frac{1}{\lambda}d_2^\epsilon(\yvec)+\frac{1}{\eta}D_2^{\epsilon}(\yvec|\tilde{\yvec})\right\},
\end{align}
can be jointly expressed as:

\[
\hat\zvec=\arg\min\limits_{\zvec\in\X\times\Y}\left\{ G(\zvec_0)^\top \zvec+\frac{1}{\lambda}H^\epsilon(\zvec)+\frac{1}{\eta}D^{\epsilon}(\zvec|\tilde{\zvec})\right\},
\]

where we define $H^{\epsilon}(\zvec):=d^{\epsilon}_1(\xvec)+d^{\epsilon}_2(\yvec)$, $D^\epsilon(\zvec|\tilde\zvec):=D_1^\epsilon(\xvec|\tilde\xvec)+D_2^\epsilon(\yvec|\tilde\yvec)$ and $G(\zvec):=(-\Umat\yvec, \Umat^\top\xvec)$.

\subsection{Convergence Analysis}
\begin{lemma}\label{lem:ne_neq1}
	Let $\zvec_{\lambda,\epsilon}^\star$ the unique NE of the game $\G_{\lambda,\epsilon}$, then for all $\zvec$:
	\[
	G(\zvec)^\top(\zvec_{\lambda,\epsilon}^\star-\zvec)+\frac{1}{\lambda}\left[H^\epsilon(\zvec_{\lambda,\epsilon}^\star)-H^\epsilon(\zvec)\right]\le 0
	\]
\end{lemma}

\begin{proof}
	Since $\zvec_{\lambda,\epsilon}^\star$ is an equilibrium of $\G_{\lambda,\epsilon}$ we have that:
	\begin{equation*}
		f_{\lambda,\epsilon}(\xvec_{\lambda,\epsilon}^\star,\yvec_{\lambda,\epsilon}^\star)\le f_{\lambda,\epsilon}(\xvec_{\lambda,\epsilon}^\star, \yvec),\quad\forall\yvec\in\Y
	\end{equation*}
	which implies that:
	\begin{equation}\label{eq:lemma21}
		\xvec_{\lambda,\epsilon}^{\star,\top}\Umat(\yvec^\star_{\lambda,\epsilon}-\yvec)+\frac{1}{\lambda}\left[d^{\epsilon}_2(\yvec^\star_{\lambda,\epsilon})-d^{\epsilon}_2(\yvec)\right]\le 0.
	\end{equation}
	Similarly we can observe that:
	\begin{equation*}
		f_{\lambda,\epsilon}(\xvec,\yvec_{\lambda,\epsilon}^\star)\le f_{\lambda,\epsilon}(\xvec_{\lambda,\epsilon}^\star, \yvec^\star_{\lambda,\epsilon}),\quad\forall\xvec\in\X
	\end{equation*}
	which implies that:
	\begin{equation}\label{eq:lemma22}
		(\xvec-\xvec^\star_{\lambda,\epsilon})^{\top}\Umat\yvec^\star_{\lambda,\epsilon}+\frac{1}{\lambda}\left[d^{\epsilon}_1(\xvec^\star_{\lambda,\epsilon})-d^{\epsilon}_1(\xvec)\right]\le 0.
	\end{equation}
	By summing Equation~\eqref{eq:lemma21} and Equation~\eqref{eq:lemma22} and observing that $\xvec^\top\Umat\yvec_{\lambda,\epsilon}^\star-\xvec_{\lambda,\epsilon}^{\star,\top}\Umat\yvec=G(\zvec)^\top(\zvec^\star_{\lambda,\epsilon}-\zvec)$ we can conclude the statement of the lemma.
\end{proof}

\begin{lemma} \label{lem:onestepreg1}
	If $\hat\zvec=\arg\min\limits_{\zvec\in \X\times\Y}\left\{ G(\zvec_0)^\top \zvec+\frac{1}{\lambda}H^\epsilon(\zvec)+\frac{1}{\eta}D^{\epsilon}(\zvec|\tilde{\zvec})\right\}$
	then:
	\begin{align}
		\eta (\hat\zvec-\zvec)^\top G(\zvec_0)+\frac{\eta}{\lambda} H^\epsilon(\hat\zvec)-\frac{\eta}{\lambda}H^\epsilon(\zvec)
		\le D^\epsilon(\zvec|\tilde{\zvec})- D^\epsilon(\hat\zvec|\tilde{\zvec})-\frac{\lambda+\eta}{\lambda}D^\epsilon(\zvec|\hat\zvec)
	\end{align}
\end{lemma}

\begin{proof}
	By the fact that for any Bregman divergence we have $\nabla_{\zvec} D^\epsilon(\zvec|\tilde{\zvec})=\nabla H^{\epsilon}(\zvec)-\nabla H^\epsilon(\tilde{\zvec})$, and rearranging the optimality condition of $\hat{\zvec}$ we get that:
	\[
	\left[G(\zvec_0)+\left(\frac{1}{\lambda}+\frac{1}{\eta}\right) \nabla H^{\epsilon}(\hat{\zvec})-\frac{1}{\eta}\nabla H^{\epsilon}(\tilde{\zvec})\right]^\top (\zvec-\hat{\zvec})\ge0,
	\]
	which implies that:
	\begin{equation}\label{eq:lemma4}
		\frac{\lambda \eta}{\lambda +\eta} G(\zvec_0)^\top(\hat{\zvec}-\zvec)\le \left[\nabla H^{\epsilon}(\hat{\zvec})-\frac{\lambda}{\lambda+\eta}\nabla H^{\epsilon}(\tilde{\zvec})\right]^\top(\zvec-\hat{\zvec}).
	\end{equation}
	It is then easy to check that for the right hand side of Equation~\eqref{eq:lemma4} the following equality     holds:
	\[
	\left[\nabla H^{\epsilon}(\hat{\zvec})-\frac{\lambda}{\lambda+\eta}\nabla H^{\epsilon}(\tilde{\zvec})\right]^\top(\zvec-\hat{\zvec})=\frac{\lambda}{\lambda+\eta}D^{\epsilon}(\zvec|\tilde{\zvec})-\frac{\lambda}{\lambda+\eta} D^{\epsilon}(\hat{\zvec}|\tilde{\zvec})+D^{\epsilon}(\zvec|\hat{\zvec})+\left[\frac{\lambda}{\lambda+\eta}-1\right]H^{\epsilon}(\hat{\zvec})-\left[\frac{\lambda}{\lambda+\eta}-1\right]H^{\epsilon}({\zvec}),
	\]
	which concludes the proof by rearranging the terms.
\end{proof}

Next we prove a similar result to~\citet[Lemma~11]{wei2020linear} that relates the distance between the gradients of two updates to the distance in the outputs. Formally:
\begin{lemma}\label{lem:gradientlemma}
	Let
	\[
	\hat\zvec_1=\arg\min\limits_{\zvec\in\X\times\Y}\left\{ G(\zvec_{0,1})^\top \zvec+\frac{1}{\lambda}H^\epsilon(\zvec)+\frac{1}{\eta}D^\epsilon(\zvec\mid\zvec_1)\right\}
	\]
	and
	\[
	\hat\zvec_2=\arg\min\limits_{\zvec\in\X\times\Y}\left\{  G(\zvec_{0,2})^\top \zvec+\frac{1}{\lambda}H^\epsilon(\zvec)+\frac{1}{\eta}D^\epsilon(\zvec\mid\zvec_1)\right\}
	\]
	then
	\[
	\|\hat\zvec_{1}-\hat\zvec_2\|_2\le \frac{\eta\lambda}{\eta+\lambda}\|G(\zvec_{0,1})-G(\zvec_{0,2})\|_2
	\]
\end{lemma}

\begin{proof}
	By summing the first order conditions of the two equations above we obtain:
	\begin{align}
		&\eta\left[G(\zvec_{0,1})-G(\zvec_{0,2})\right]^\top(\hat\zvec_2-\hat\zvec_1)\ge\left(1+\frac{\eta}{\lambda}\right)\left[\nabla H^\epsilon(\hat\zvec_1)-\nabla H^\epsilon(\hat\zvec_2)\right]^\top(\hat\zvec_1-\hat\zvec_2).
	\end{align}
	By strong convexity of $H^\epsilon(\cdot)$, given by Lemma~\ref{lem:strong_convexity2}, the right hand side is lower bounded by $\left(1+\frac{\eta}{\lambda}\right)\|\hat\zvec_1-\hat\zvec_2\|_2^2$, while by Cauchy-Schwartz inequality we have that the wight hand side is upper bounded by $\eta\|G(\zvec_{0,1})-G(\zvec_{0,2})\|_2\|\hat\zvec_1-\hat\zvec_1\|_2$.
	Thus:
	\[
	\eta\|G(\zvec_{0,1})-G(\zvec_{0,2})\|_2\|\hat\zvec_1-\hat\zvec_1\|_2\ge \left(1+\frac{\eta}{\lambda}\right)\|\hat\zvec_1-\hat\zvec_2\|_2^2,
	\] 
	which, after dividing by $\|\hat\zvec_1-\hat\zvec_2\|_2$, concludes the proof.
\end{proof}

\linearconv*
\begin{proof}
	Notice that the update of Algorithm~\ref{alg:OOMD} can be jointly written as:
	\begin{align*}
		\zvec^{(k)}_{t+1}=\arg\min\limits_{\zvec}\left\{\zvec^\top G(\zvec^{(k)}_t)+\frac{1}{\lambda_k}H^{\epsilon_k}(\zvec) + \frac{1}{\lambda_k}D^{\epsilon_k}\left(\zvec|\zvec^{(k)}_{t+\frac{1}{2}}\right) \right\},\\
		\zvec^{(k)}_{t+\frac{1}{2}}=\arg\min\limits_{\zvec}\left\{\zvec^\top G(\zvec^{(k)}_t)+\frac{1}{\lambda_k}H^{\epsilon_k}(\zvec) + \frac{1}{\lambda_k}D^{\epsilon_k}\left(\zvec|\zvec^{(k)}_{t-\frac{1}{2}}\right) \right\}.
	\end{align*}
	
	Then we can use Lemma~\ref{lem:onestepreg1} for the update of $\zvec^{(k)}_{t+\frac{1}{2}}$ which gives:
	\begin{align}\label{eq:uno_onestep_reg}
		&\eta (\zvec^{(k)}_{t+\frac{1}{2}}-\zvec)^\top G(\zvec_t)+\frac{\eta}{\lambda_k} H^{\epsilon_k}(\zvec^{(k)}_{t+\frac{1}{2}})-\frac{\eta}{\lambda_k}H^{\epsilon_k}(\zvec)\le D^{\epsilon_k}(\zvec|\zvec^{(k)}_{t-\frac{1}{2}})- D^{\epsilon_k}(\zvec^{(k)}_{t+\frac{1}{2}}|\zvec^{(k)}_{t-\frac{1}{2}})-\frac{\lambda_k+\eta}{\lambda_k}D^{\epsilon_k}(\zvec|\zvec^{(k)}_{t+\frac{1}{2}}),
	\end{align}
	which holds for any $\zvec$ and thus also for $\zvec_k^\star$.
	
	On the other hand using Lemma~\ref{lem:onestepreg1} for the update of $\zvec^{(k)}_{t}$ which gives:
	\begin{align}\label{eq:due_onestep_reg}
		&\eta (\zvec^{(k)}_{t}-\zvec)^\top G(\zvec^{(k)}_{t-1})+\frac{\eta}{\lambda_k} H^{\epsilon_k}(\zvec^{(k)}_{t})-\frac{\eta}{\lambda_k}H^{\epsilon_k}(\zvec) \le D^{\epsilon_k}(\zvec|\zvec^{(k)}_{t-\frac{1}{2}})- D^{\epsilon_k}(\zvec^{(k)}_{t}|\zvec_{t-\frac{1}{2}})-\frac{\lambda_k+\eta}{\lambda_k}D^{\epsilon_k}(\zvec|\zvec^{(k)}_{t}),
	\end{align}
	which holds fro any $\zvec$ and thus also for $\zvec^{(k)}_{t+\frac{1}{2}}$.
	Summing Equation~\eqref{eq:uno_onestep_reg} with $\zvec=\zvec_k^\star$ to Equation~\eqref{eq:due_onestep_reg} with $\zvec=\zvec^{(k)}_{t+\frac{1}{2}}$, and rearranging results in the following inequality:
	\begin{align}
		\eta(\zvec^{(k)}_t-\zvec^\star)^\top G(\zvec^{(k)}_t)&+\frac{\eta}{\lambda_k}\left[H^{\epsilon_k}(\zvec^{(k)}_t)-H^{\epsilon_k}(\zvec^\star)\right]\nonumber\\
		&\le D^{\epsilon_k}(\zvec_k^\star|\zvec^{(k)}_{t-\frac{1}{2}})-\frac{\lambda_k+\eta}{\lambda_k} D^{\epsilon_k}(\zvec_k^\star|\zvec^{(k)}_{t+\frac{1}{2}})-D^{\epsilon_k}(\zvec^{(k)}_t|\zvec^{(k)}_{t-\frac{1}{2}})-\frac{\lambda_k+\eta}{\lambda_k} D^{\epsilon_k}(\zvec^{(k)}_{t+\frac{1}{2}}|\zvec^{(k)}_{t})\nonumber\\
		&+\eta(\zvec^{(k)}_t-\zvec^{(k)}_{t+\frac{1}{2}})^\top(G(\zvec^{(k)}_t)-G(\zvec^{(k)}_{t-1}))\label{eq:th4_1}
	\end{align}
	Now consider the last term of Equation~\eqref{eq:th4_1} and the following inequalities:
	\begin{align}
		(\zvec^{(k)}_t-\zvec^{(k)}_{t+\frac{1}{2}})^\top(G(\zvec^{(k)}_t)-G(\zvec^{(k)}_{t-1}))&\le\|\zvec^{(k)}_t-\zvec^{(k)}_{t+\frac{1}{2}}\|_2\|G(\zvec^{(k)}_t)-G(\zvec^{(k)}_{t-1})\|_2\\
		&\le\frac{\lambda_k\eta}{\lambda_k+\eta}\|G(\zvec^{(k)}_{t-1})-G(\zvec^{(k)}_t)\|_2^2\\
		&\le \frac{\lambda_k\eta}{\lambda_k+\eta}L_{\Umat}^2\|\zvec^{(k)}_{t-1}-\zvec^{(k)}_t\|_2^2\\
		&\le \frac{2\lambda_k\eta}{\lambda_k+\eta}L_{\Umat}^2\left[\|\zvec^{(k)}_{t-1}-\zvec^{(k)}_{t-\frac{1}{2}}\|_2^2+\|\zvec^{(k)}_{t-\frac{1}{2}}-\zvec^{(k)}_{t-1}\|^2_2\right]\\
		&\le \frac{4\lambda_k\eta}{\lambda_k+\eta}L_{\Umat}^2\left[D^{\epsilon_k}(\zvec^{(k)}_t|\zvec^{(k)}_{t-\frac{1}{2}})+D^{\epsilon_k}(\zvec^{(k)}_{t-\frac{1}{2}}|\zvec^{(k)}_{t-1})\right],
	\end{align}
	where the first inequality is the Cauchy Schwartz inequality and the second inequality follows from Lemma~\ref{lem:gradientlemma}. Then we used the fact that the operator $G(\cdot)$ is linear with matrix $\Amat$ defined as:
	\[
	\Amat=\left[ {\begin{array}{cc}
			\zerovec & -\Umat^\top \\
			\Umat& \zerovec
	\end{array} } \right],
	\]
	and $L_{\Umat}\coloneqq\|\Umat\|_2=\|\Amat\|_2$.
	Thus, continuing from Equation~\eqref{eq:th4_1} we get:
	\begin{align*}
		&\eta(\zvec^{(k)}_t-\zvec^\star_k)^\top G(\zvec^{(k)}_t)+\frac{\eta}{\lambda_k}\left[H^{\epsilon_k}(\zvec^{(k)}_t)-H^{\epsilon_k}(\zvec_k^\star)\right]\\
		&\le D^{\epsilon_k}(\zvec_k^\star|\zvec^{(k)}_{t-\frac{1}{2}})-\frac{\lambda_k+\eta}{\lambda_k} D^{\epsilon_k}(\zvec_k^\star|\zvec^{(k)}_{t+\frac{1}{2}})-D^{\epsilon_k}(\zvec^{(k)}_t|\zvec^{(k)}_{t-\frac{1}{2}})-\frac{\lambda_k+\eta}{\lambda_k} D^{\epsilon_k}(\zvec^{(k)}_{t+\frac{1}{2}}|\zvec^{(k)}_{t})\\
		&+\frac{4\lambda_k\eta^2}{\lambda_k+\eta}L_{\Umat}^2\left[D^{\epsilon_k}(\zvec^{(k)}_t|\zvec^{(k)}_{t-\frac{1}{2}})+D^{\epsilon_k}(\zvec^{(k)}_{t-\frac{1}{2}}|\zvec^{(k)}_{t-1})\right]
	\end{align*}
	By assuming that $\lambda_k\ge\eta$ and $\eta\le\frac{1}{\sqrt{2}L_{\Umat}}$ we have that  $\frac{4\lambda_k\eta^2}{\lambda_k+\eta}L_{\Umat}^2\le 1$. Moreover, thanks to Lemma~\ref{lem:ne_neq1} we have that:
	\[
	\eta(\zvec^{(k)}_t-\zvec_k^\star)^\top G(\zvec^{(k)}_t)+\frac{\eta}{\lambda_k}\left[H^{\epsilon_k}(\zvec^{(k)}_t)-H^{\epsilon_k}(\zvec_k^\star)\right]\ge 0.
	\]
	Thus, by rearranging Equation\eqref{eq:th4_1} and recalling that $D^{\epsilon_k}(\cdot|\cdot)\ge 0$ we obtain:
	\[
	\frac{\lambda_k+\eta}{\lambda_k}\left[D^{\epsilon_k}(\zvec_k^\star|\zvec^{(k)}_{t+\frac{1}{2}})+D^{\epsilon_k}(\zvec^{(k)}_{t+\frac{1}{2}}|\zvec^{(k)}_t)\right]\le D^{\epsilon_k}(\zvec_k^\star|\zvec^{(k)}_{t-\frac{1}{2}})+D^{\epsilon_k}(\zvec^{(k)}_{t-\frac{1}{2}}|\zvec^{(k)}_{t-1}).
	\]
	Thus, by iterating the above expression, we get that:
	\[
	D^{\epsilon_k}(\zvec_k^\star|\zvec^{(k)}_{t+\frac{1}{2}})+D^{\epsilon_k}(\zvec^{(k)}_{t+\frac{1}{2}}|\zvec^{(k)}_t)\le \left(\frac{\lambda_k}{\lambda_k+\eta}\right)^{t}D^{\epsilon_k}(\zvec_k^\star|\zvec^{(k)}_{0}),
	\]
	and since $D^{\epsilon_k}(\zvec_1,\zvec_2)\ge\frac{1}{2}\|\zvec_1-\zvec_2\|_2^2\ge0$ we have that:
	\[
	\|\zvec_k^\star-\zvec^{(k)}_{t+\frac{1}{2}}\|_2^2\le2 \left(\frac{\lambda_k}{\lambda_k+\eta}\right)^{t}D^{\epsilon_k}(\zvec_k^\star|\zvec^{(k)}_{0}),
	\]
	and
	\[
	\|\zvec^{(k)}_t-\zvec^{(k)}_{t+\frac{1}{2}}\|_2^2\le2 \left(\frac{\lambda_k}{\lambda_k+\eta}\right)^{t}D^{\epsilon_k}(\zvec_k^\star|\zvec^{(k)}_{0}),
	\]
	Thus we can conclude that:
	\begin{align}
		\|\zvec_k^\star-\zvec^{(k)}_t\|_2&\le\|\zvec_k^\star-\zvec^{(k)}_{t+\frac{1}{2}}\|_2+	\|\zvec^{(k)}_t-\zvec^{(k)}_{t+\frac{1}{2}}\|_2\le2\sqrt{2D^{\epsilon_k}(\zvec_k^\star|\zvec^{(k)}_0)}\left(\frac{\lambda_k}{\lambda_k+\eta}\right)^{t/2}.
	\end{align}
which concludes the proof.
\end{proof}

\subsection{Exploitability }

In this section we will prove all the results needed to prove the convergence rate of Algorithm~\ref{alg:OOMD} in terms of Nash gap.

We first need the following lemma.

\begin{lemma}\label{lem:inequality}
	For all $x\ge a>0$ we have the following inequality: 
	\[
	\left(\frac{x}{a+x}\right)^{x^2/a}\le\frac{2}{x}.
	\]
\end{lemma}

\begin{proof}
	
	Consider $(1-y)^r$ with $r>0$ and $y<1$ and define $b=\log\left(\frac{1}{1-y}\right)$. Then the following inequality holds:
	\begin{align}\label{eq:lemma81}
		(1-y)^r=e^{-rb}\le\frac{1}{1+rb}\le\frac{1}{1+r\log\left(\frac{1}{1-y}\right)}.
	\end{align}
	
	Define $\phi=a+x$. Thanks to Equation~\eqref{eq:lemma81} we get:
	\begin{align*}
		\left(\frac{x}{a+x}\right)^{x^2/a}=\left(1-\frac{a}{\phi}\right)^\frac{(\phi-a)^2}{a}\le\frac{1}{1+\frac{(\phi-a)^2}{a}\log\left(\frac{\phi}{\phi-a}\right)}.
	\end{align*}
	Substituting back $x=\phi-a$ we can notice that $\log(1+a/x)\ge \frac{a/x}{1-a/x}=\frac{a}{x-a}\ge \frac{a}{2x}$, where the last inequality follows from $x\ge a$. Thus we obtain:
	\[
	\left(\frac{x}{a+x}\right)^{x^2/a}\le\frac{1}{1+\frac{x^2}{a}\log\left(1+\frac{a}{x}\right)}\le\frac{1}{1+\frac{x^2}{a}\frac{a}{x-a}}=\frac{1}{1+x/2}\le\frac{2}{x},
	\]
	thus concluding the proof.
\end{proof}

\anytimeconvtwo*

\begin{proof}
	Consider now the following chain of inequalities:
	\[
	\|\zvec_k^\star-\zvec_{T_k}^{(k)}\|_2\stackrel{(i)}{=}O\left(\left(\frac{\lambda_k}{\lambda_k+\eta}\right)^{T_k/2}\right)\stackrel{(ii)}{=}O({1/\lambda_k}),
	\]
	where $(i)$ is due to Theorem~\ref{thm:algo_conv} and $(ii)$ is due to Lemma~\ref{lem:inequality} and the assumption that $T_k:=\beta^k\ge\frac{2\lambda_k^2}{\eta}$.
	Thus $\lim_{k\to\infty}\zvec_{T_k}^{(k)}=\zvec_k^\star$.
	Moreover from Theorem~\ref{th:convergece} we know that $\zvec_k^\star\to\zvec^\star$ and thus by triangular inequality we can show that:
	\[
	\|\zvec^\star-\zvec_{T_k}^{(k)}\|_2\le\|\zvec^\star-\zvec_k^\star\|_2+\|\zvec_k^\star-\zvec_{T_k}^{(k)}\|_2,
	\]
	and both terms on the right hand side goes to zero with $k\to\infty$, thus concluding the proof.
\end{proof}

\saddelpointgap*

\begin{proof}
	We will first consider the saddle point gap with respect the utility function $f_{\lambda,\epsilon}$ of the regularized-perturbed game:
	\begin{align}
		f_{\lambda,\epsilon}(\xvec,\tilde\yvec)-f_{\lambda,\epsilon}(\tilde\xvec,\yvec)\coloneqq &\xvec^\top \Umat\tilde\yvec-\frac{1}{\lambda}d_1^\epsilon(\xvec)+\frac{1}{\lambda}d_2^\epsilon(\tilde\yvec)-\tilde{\xvec}^\top\Umat\yvec+\frac{1}{\lambda}d_1^\epsilon(\tilde\xvec)-\frac{1}{\lambda}d_2^\epsilon(\yvec)
	\end{align}
	Now we can add and subtract the following quantities $\xvec^\top\Umat\yvec^\star_{\lambda,\epsilon}, \xvec^{\star,\top}_{\lambda,\epsilon}\Umat\yvec, \frac{1}{\lambda}d_1^\epsilon(\xvec^\star_{\lambda,\epsilon})$ and $\frac{1}{\lambda}d_2^\epsilon(\yvec^\star_{\lambda,\epsilon})$ to obtain:
	\begin{align}
		f_{\lambda,\epsilon}(\xvec,\tilde\yvec)-f_{\lambda,\epsilon}(\tilde\xvec,\yvec)=&\underbrace{\xvec^\top\Umat(\tilde\yvec-\yvec^\star_{\lambda,\epsilon})-(\tilde\xvec-\xvec^\star_{\lambda, \epsilon})^\top\Umat\yvec}_{\text{\mycircled{1}}}+\underbrace{\frac{1}{\lambda}d_1^\epsilon(\tilde\xvec)-\frac{1}{\lambda}d_1^\epsilon(\xvec^\star_{\lambda,\epsilon})+\frac{1}{\lambda}d_2^\epsilon(\tilde\yvec)-\frac{1}{\lambda}d_2^\epsilon(\yvec^\star_{\lambda,\epsilon})}_{\text{\mycircled{2}}}\\
		&+\underbrace{f_{\lambda,\epsilon}(\xvec,\yvec^\star_{\lambda, \epsilon})-f_{\lambda,\epsilon}(\xvec^\star_{\lambda, \epsilon},\yvec)}_{\text{\mycircled{3}}}.
	\end{align}
	We can observe that by Cauchy-Schwartz inequality we have that the \mycircled{1} is can be upper bounded by:
	\begin{align}
		\nu L_{\Umat} (\|\xvec\|_2+\|\yvec\|_2)&\le \nu L_{\Umat}\sqrt{2|\Sigma_1|+2|\Sigma_2|}\\
		&\le 2\nu |\Sigma|^{3/2},
	\end{align}
	
	where $L_{\Umat}:=\|\Umat\|_2\le\sqrt{|\Sigma_1|\cdot|\Sigma_2|}\le|\Sigma|$.

	On the other hand \mycircled{2} is upper bounded by $4C/\lambda$ (see Lemma~\ref{lem:bounded_d} for the definition of the constant $C$). Finally \mycircled{3} is negative by definition of the equilibria $\zvec_{\lambda,\epsilon}^\star=(\xvec_{\lambda,\epsilon}^\star, \yvec_{\lambda,\epsilon}^\star)$.
	Thus we have:
	\[
	f_{\lambda,\epsilon}(\xvec,\tilde\yvec)-f_{\lambda,\epsilon}(\tilde\xvec,\yvec)\le\nu|\Sigma|^{3/2}+\frac{4C}{\lambda}.
	\]
	On the other hand, following the same argument as in the proof of Theorem~\ref{thm:connection_efpe} we know that $|f_{\lambda,\epsilon}(\xvec,\yvec)-\xvec^\top\Umat\yvec|\le\frac{2C}{\lambda}$ which directly implies that: 
	\[
	\xvec^\top\Umat\tilde\yvec-\tilde{\xvec}^\top\Umat\yvec\le\nu |\Sigma|^{3/2}+\frac{8C}{\lambda}.
	\]
\end{proof}

\matrixgapfinal*

\begin{proof}
	From Theorem~\ref{thm:algo_conv} we have that at the end of the $k$-th phase of Algorithm~\ref{alg:OOMD} we can guarantee:
	\[
	\|\zvec_k^\star-\zvec^{(k)}_{T_k}\|_2\le2\sqrt{2D^{\epsilon_k}\left( \zvec_k^\star|\zvec^{(k)}_0 \right)}\left(\frac{\lambda_k}{\lambda_k+\eta}\right)^{T_k/2}.
	\]
	Moreover, by the assumption on $\lambda_k$ that $T_k:=\beta^k\ge\frac{2\lambda_k^2}{\eta}$, and combining Theorem~\ref{thm:algo_conv} with Lemma~\ref{lem:inequality} we have that:
	\[
	\left(\frac{\lambda_k}{\lambda_k+\eta}\right)^{T_k/2}\le\frac{2}{\lambda_k}
	\] 
	This, thanks to Theorem~\ref{thm:algo_conv}, let us conclude that:
	\[
		\|\zvec_k^\star-\zvec^{(k)}_{T_k}\|_2\le \frac{4}{\lambda_k}\left[2D^{\epsilon_k}(\zvec_{k}^\star|\zvec_0^{(k)})\right]^{1/2}
	\]
	and thanks to Theorem~\ref{thm:saddelpointgap} we can readily conclude that:
	\[
	\xvec^\top\Umat\yvec^{(k)}_{T_k}-{\xvec}^{(k),\top}_{T_k}\Umat\yvec\le \frac{4}{\lambda_k}|\Sigma|^{3/2}\left[2D^{\epsilon_k}(\zvec_{k}^\star|\zvec_0^{(k)})\right]^{1/2} +\frac{8C}{\lambda_k}.
	\]

\end{proof}

\subsection{Decomposition}

Here we prove the results on the decomposition of Algorithm~\ref{alg:OOMD} on the tree. We also report here for completeness of exposition the decomposition of in~\citet{farina2019optimistic,farina2021better} that shows how to compute the conjugate gradient $\nabla d^{\epsilon,*}_i(\cdot)$ and the gradient $\nabla d^{\epsilon}_i(\cdot)$ over the entire tree by computing local conjugate gradient $\nabla d_{\Delta_I}^{\epsilon,*}(\cdot)$ and gradient $\nabla d_{\Delta_I}^{\epsilon}(\cdot)$  at each infoset $I\in\I_i$.

\begin{algorithm}[H]
	\small
	\caption{Prox-mapping decomposition~\citep{farina2021better}}
	\label{alg:decomposition}
	\begin{algorithmic}[1]
		\Function{$\nabla d^{\epsilon,*}_i$}{$\gvec$}\Comment{\textcolor{gray}{Conjugate gradient computation.}}
		\State $\vvec\gets\zerovec\in\Reals^{|\Sigma_i|}$
		\State $\vvec[\varnothing]\gets 1$
		\For{$I\in\I_i$ in bottom-up order}\Comment{\textcolor{gray}{Compute the behavioral strategy $\vvec$ in bottom-up fashion.}}
			\State $\vvec[\sigma_i(I)a]\gets\nabla d_{\Delta_I}^{\epsilon,*}(\gvec[I]/\alpha_I)[a]$
			\State $\gvec[\sigma_I] \gets \gvec[\sigma_I]-\alpha_Id_{\Delta_I}^\epsilon(\vvec[I])-\gvec[I]^\top\vvec[I]$
		\EndFor
		\For{$I\in\I_i$ in top-down order}\Comment{\textcolor{gray}{Convert the strategy $\vvec\in\Reals^{|\Sigma_i|}$ from behavioral to sequence form.}}
			\State $\vvec[\sigma(I)a]\gets \vvec[\sigma(I)a]\vec[\sigma(I)]$
		\EndFor
		\State \Return $\vvec$
		\EndFunction		%
	\\\hrulefill
	\Function{$\nabla d^{\epsilon}_i$}{$\gvec$}\Comment{\textcolor{gray}{Gradient computation.}}
		\State $\vvec\gets\zerovec\in\Reals^{|\Sigma_i|}$
		\For{$I\in\I_i$ in bottom-up order}\Comment{\textcolor{gray}{Compute the gradient $\vvec$ in bottom-up fashion.}}
		\State $\vvec[\sigma_i(I)a]\gets\vvec[\sigma_i(I)a]+\alpha_I\nabla d_{\Delta_I}^\epsilon\left(\frac{\gvec[\sigma_i(I)a]}{\gvec[\sigma_i(I)]}\right)$
		\State $\vvec[\sigma_i(I)]\gets\vvec[\sigma_i(I)]+\alpha_I d_{\Delta_I}^\epsilon\left(\frac{\gvec[\sigma_i(I)a]}{\gvec[\sigma_i(I)]}\right)-\alpha_I\left(\frac{\gvec[\sigma_i(I)a]}{\gvec[\sigma_i(I)]}\right)^\top\nabla d_{\Delta_I}^\epsilon\left(\frac{\gvec[\sigma_i(I)a]}{\gvec[\sigma_i(I)]}\right)$
		\EndFor
		\State \Return $\vvec$
	\EndFunction	
	\end{algorithmic}%
\end{algorithm}

\efficencyone*

\begin{proof}
	\textbf{(i)} Let us firs prove the first statement, and let us consider the following chain of equations:
	\begin{align}
		\hat{\xvec}&:=\arg\max\limits_{\xvec\in\X}\left\{\xvec^\top\gvec-\frac{d_1^\epsilon(\xvec)}{\lambda}-\frac{1}{\eta}D_1^\epsilon(\xvec|\tilde{\xvec})\right\}\\
		&=\arg\max\limits_{\xvec\in\X}\left\{\xvec^\top\gvec-\left(\frac{1}{\eta}+\frac{1}{\lambda}\right)d_1^\epsilon(\xvec)+\frac{1}{\eta}d_1^\epsilon(\tilde\xvec)+\frac{1}{\eta}\nabla d_1^\epsilon(\tilde\xvec)^\top(\xvec-\tilde\xvec)\right\}\\
		&=\arg\max\limits_{\xvec\in\X}\left\{\xvec^\top\left(\gvec+\frac{1}{\eta}\nabla d_1^\epsilon(\tilde\xvec)\right)-\left(\frac{1}{\eta}+\frac{1}{\lambda}\right)d_1^\epsilon(\xvec)\right\}\\
		&=\arg\max\limits_{\xvec\in\X}\left\{\xvec^\top\left(\gamma\gvec+\frac{\gamma}{\eta}\nabla d_1^\epsilon(\tilde\xvec)\right)-d_1^\epsilon(\xvec)\right\}\\
		&\coloneqq\nabla d_1^{\epsilon,*}(\tilde\gvec),
	\end{align}
	where we defined $\tilde\gvec:=\gamma\gvec+\frac{\gamma}{\eta}\nabla d_1^\epsilon(\tilde\xvec)$ and $\frac{1}{\eta}+\frac{1}{\lambda}=\frac{1}{\gamma}$, which concludes the first statement whenever $\nabla d_1^\epsilon(\tilde\xvec)$ is efficiently computable, which, thanks to Algorithm~\ref{alg:decomposition} happens whenever the local gradient $\nabla d_{\Delta_I}^\epsilon(\tilde\xvec)$ have a closed formula (see below the proof for statement \textbf{(ii)}).

	\textbf{(ii)} Now let us turn to the second statement, which concerns the closed formula updates of the local conjugate gradients of $d_{\Delta_I}^\epsilon$.
	It is well known that $\nabla d_{\Delta_I}^{0,*}(\tilde\gvec)$ is solved by $\frac{e^{-\tilde\gvec[a]}}{\sum\limits_{b\in A(I)} e^{-\tilde\gvec[b]}}$.
	Moreover, it is straightforward to verify that:
	\begin{align}
		\nabla d_{\Delta_I}^{\epsilon,*}(\tilde\gvec)&:=\arg\max\limits_{\wvec\in\Delta_{I}} \left\{\wvec^\top\tilde\gvec-d_{\Delta_I}^\epsilon(\wvec)\right\}\\
		&=\boldsymbol{1}\epsilon+\arg\max\limits_{\substack{\wvec\ge 0,\\ \sum_{k=1}^{n_I}\wvec[k]=1-\epsilon\cdot n_I}} \left\{\wvec^\top\tilde\gvec-d_{\Delta_I}^0(\wvec)\right\}\,
	\end{align}
	where we used the change of variable $\wvec \mapsto \wvec-\boldsymbol{1}\epsilon$.
	Clearly a similar statement holds for the updates of the second player.
	
	Now we employ the change of variable $\wvec\mapsto\wvec/(1-\epsilon \cdot n_I)$ which gives:
	\begin{align}
		\nabla d_{\Delta_I}^{\epsilon,*}(\tilde\gvec)&=\boldsymbol{1}\epsilon+(1-\epsilon \cdot n_I)\arg\max\limits_{\wvec\in \Delta_{I}}\left\{\wvec^\top\tilde\gvec\nabla -d_{\Delta_I}^0(\wvec)\right\}\\
		&=\boldsymbol{1}\epsilon+(1-\epsilon \cdot n_I)\arg\max\limits_{\wvec\in \Delta_{I}}\nabla d_{\Delta_I}^{0,*}(\tilde\gvec)\\
		&=\boldsymbol{1}\epsilon+(1-\epsilon \cdot n_I)\frac{e^{\tilde\gvec}}{\sum\limits_{b\in A(I)}e^{\tilde\gvec[b]}},
	\end{align}
	which proves the second statement.
	
	Finally the computation of the gradient of $d^\epsilon_{\Delta_I}(\wvec):=(\wvec-\onevec\epsilon)^\top\log( \wvec-\onevec\epsilon)$ trivially follows from direct computation.

\end{proof}

\section{Additional Experiments}\label{app:3}

\subsection{Games Description}

\begin{wrapfigure}{l}{0.36\textwidth}
	\centering
	\begin{tabular}{l|c|c}
		& $|\mathcal{I}|$ & $|\Sigma|$ \\ \noalign{\smallskip}\hlinewd{1pt}\noalign{\smallskip}
		Kuhn      & $6$             & $13$       \\
		Leduc 3   & $114$           & $337$      \\
		Leduc 5   & $390$           & $911$      \\
		Goofspiel & $57$            & $118$      \\
		dRPS      & $3$             & $10$      
	\end{tabular}
	\caption{Size of the game instances with respect to the number of infosets $|\I|$ and the number of sequences $|\Sigma|$.}
\end{wrapfigure}

First we are going to describe in details the games used in the experimental evaluation.

\paragraph{Kuhn Poker} 
Is a simplified poker game in which ~\cite{kuhn1950simplified} in which played with $3$ cards. Each player then pays one blind to the pot and is dealt a private card. The first player then decide to either check or to bet (places an additional blind on the pot). In the first case the second player can either check or to bet or fold/call in the second case. In the case in which the second player has placed a bet, the first player still has to decide weather to call or to fold. In the case no one has folded, there is a showdown phase in which the player with the hand of higher value wins the pot.

\begin{wrapfigure}{r}{0.25\textwidth}
	\begin{center}
		\includegraphics[width=0.19\textwidth]{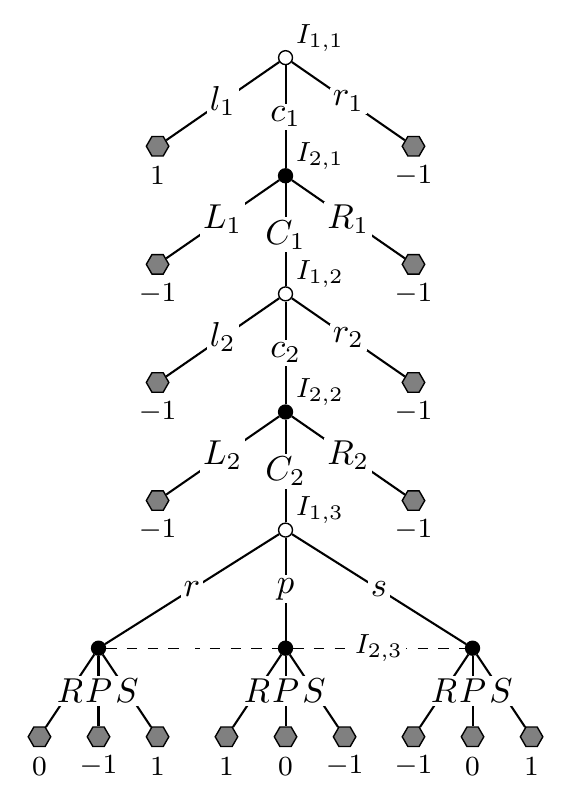}
	\end{center}
	\caption{Game tree of the dRPS game described in this section. White (black) round nodes are nodes of the infoset of the first (second) player. $I_{i,j}$ is the $i$-th player's $j$-th infoset. }
	\label{fig:dRPS}
\end{wrapfigure}

\paragraph{Leduc Poker $(n)$} Leduc is a card game played first introduced in~\cite{southey2005bayes}. It is played with $2n$ cards, where $n$ is called the rank of the game. Each player is dealt a private card an there is a common, unknown card. At the start of every hand, each player places one blind in the pot. There are two betting phases, which are identical and work as in the Kuhn poker, one before the revel of the common card, and one after. After the two betting stages, if no player folded, the player whose hand is of higher value wins the pot.

\paragraph{Goofspiel} It is a card game introduce by~\cite{ross1971goofspiel}. Each of the two players has $3$ ordered cards, which are used to privately bet on a community card (also of value from $1$ to $3$) revealed et each of the $3$ turns. The player who has bet the highest card at a specific turn, wins the amount represented by the community card. In case of ties the community card is disregarded, otherwise it has the value of the card itself.

\paragraph{dRPS} 

We designed dRPS which is a ``deep'' version of Rock Paper Scissor in which each player has $3$ actions at each infoset for two alternating times. Then, after the second move of the second player, begins a turn of Rock Paper Scissor. The payoffs of such a game are designed so that all the NE of the game prescribe the first player to end the game straight away. This renders irrelevant all the moves played on the other infoset. Indeed every strategy that assign probability $1$ to action $l_1$ of the first player is a NE. The game is designed so that at deep infosets the chance of an algorithm which is \emph{not} designed to explore the tree would have low probability of visiting such infosets due to random chance.

\paragraph{Matrix Game}

We also consider a $3\times 3$ normal form game (that was also used in the proof of Proposition~\ref{th:doublelimit}). 
The utility matrix, in which element $i,j$ is the utility received by the first player when it plays $i$ and the second player plays $2$, is the following one:
	\[
\left[ {\begin{array}{ccc}
		0.3 & 0.5 & 0.3 \\
		0.7& 0.3 & 0.7 \\
		0.6 & 0.2 & 0.2
\end{array} } \right]
\]
In this game we can easily compute the unique perfect equilibria that is the mixed strategy $\zvec^\star=\left([1/2,1/2,0]^\top,[0,2/3,1/3]^\top\right)$.

\subsection{Additional Results}

\begin{figure}[t!]
	\centering
	\begin{subfigure}{.49\columnwidth}
		\centering
		\includegraphics[width=\linewidth]{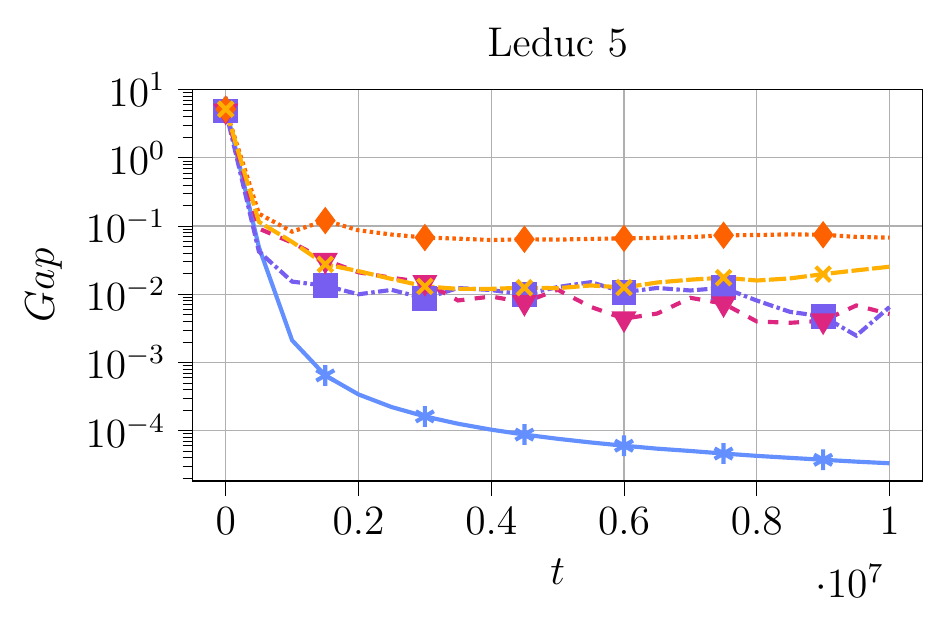}
	\end{subfigure}%
	\hfill
	\begin{subfigure}{.49\columnwidth}
		\centering
		\includegraphics[width=\linewidth]{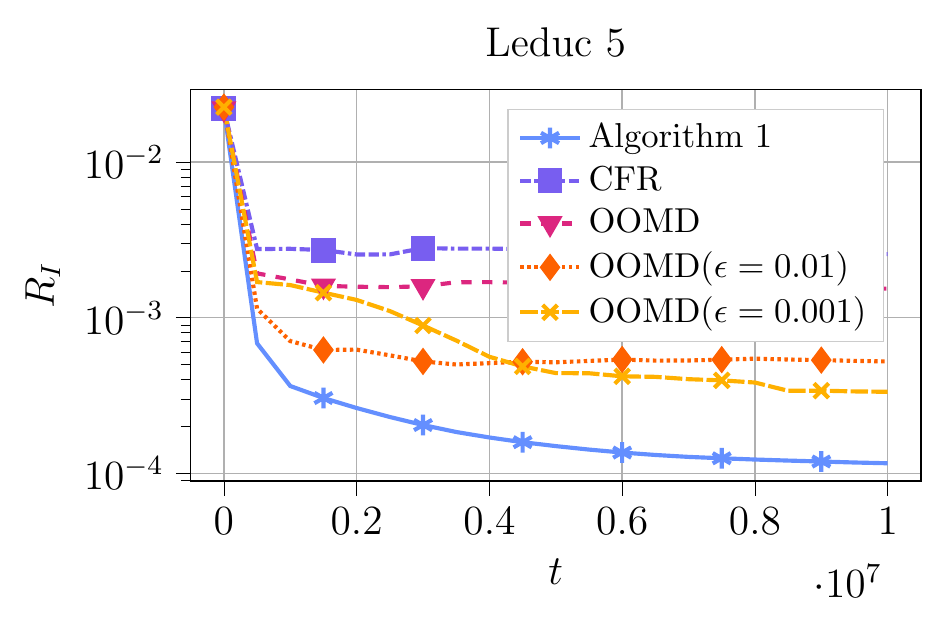}
	\end{subfigure}
	
	\begin{subfigure}{.49\columnwidth}
		\centering
		\includegraphics[width=\linewidth]{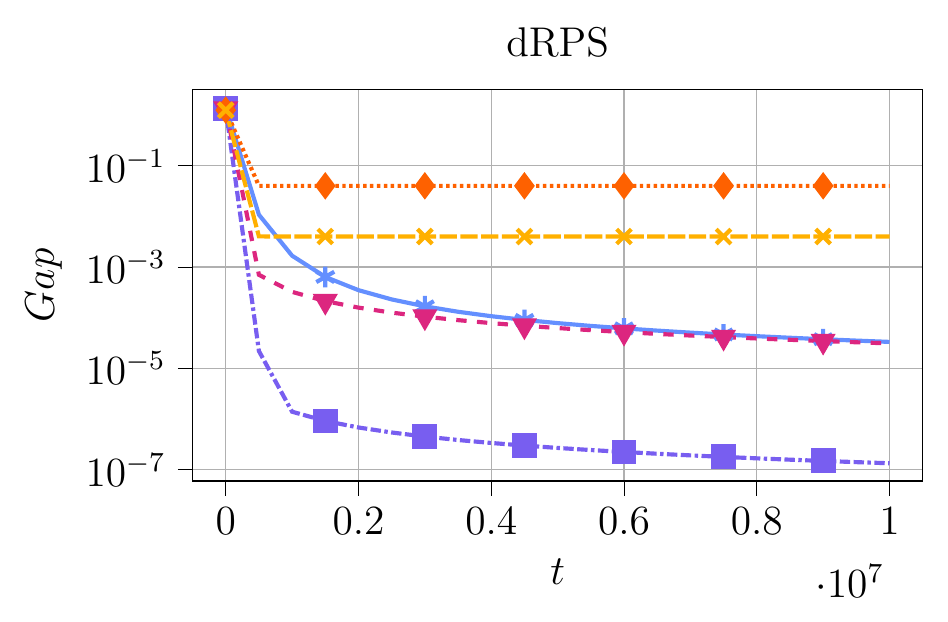}
	\end{subfigure}%
	\hfill
	\begin{subfigure}{.49\columnwidth}
		\centering
		\includegraphics[width=\linewidth]{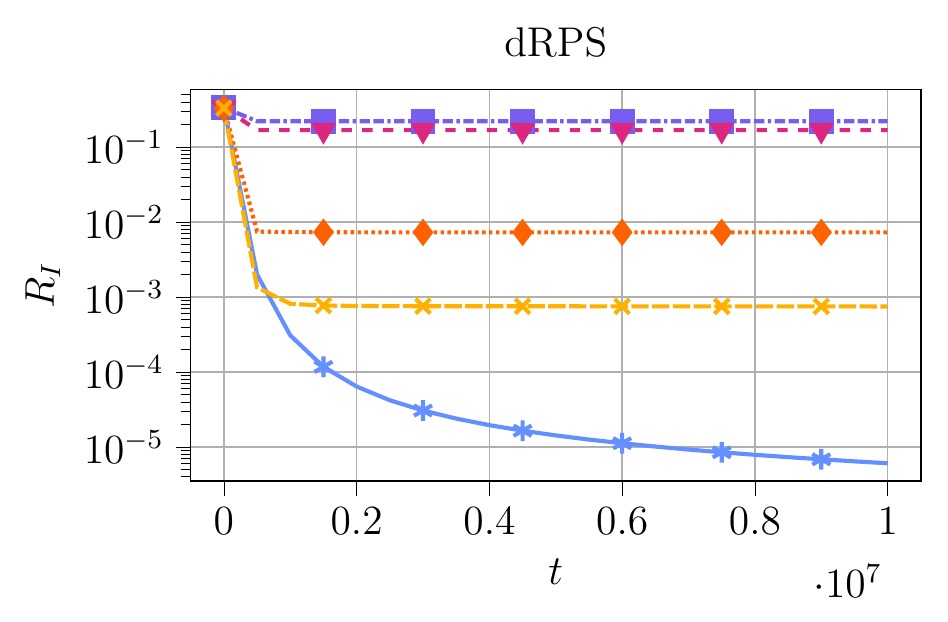}
	\end{subfigure}
	
	\caption{Results of the experimental evaluation. Algorithm~\ref{alg:OOMD} is compared with the baselines in terms of Nash gap (\emph{Left}) and average infoset regret (\emph{Right}).}
	\label{fig:exp2}
\end{figure}

\begin{wrapfigure}{r}{0.51\textwidth}
	\centering
	\includegraphics[width=0.50\textwidth]{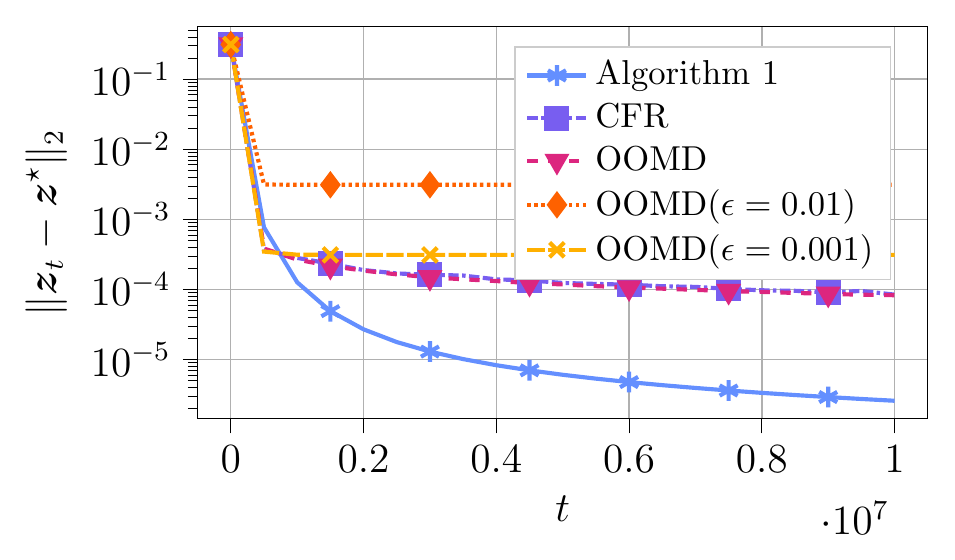}
	\caption{Distance to the unique EFPE of the Matrix Game.}
	\label{fig:distance}
\end{wrapfigure}

We can see in Figure~\ref{fig:exp2} that in both Leduc $5$ and dRPS, Algorithm~\ref{alg:OOMD} outperforms the baselines in term if average infoset regret $R_I$. This is the case for both the baselines which are agnostic to refinements (CFR and OOMD) and those who are build to work for refinements, but only find approximate ones (OOMD$(\epsilon=0.01)$ and OOMD$(\epsilon=0.001)$).
On the other hand, as seen in Section~\ref{sec:exper}, this is not always the case for what concern the Nash gap, even if we can see that in Leduc $5$, Algorithm~\ref{alg:OOMD} still manages to outperform all the baselines by $2$ orders of magnitude. However, this is not always the case. Indeed, in dRPS, CFR outperforms all the other methods by $2$ orders of magnitude.
As already remarked in Section~\ref{sec:exper}, we think that investigating such behavior requires further investigation.

Finally we used the Matrix Game introduce in the previous section, of which we know the only perfect equilibria, to compute the distance $\|\zvec_t-\zvec^\star\|_2$ of the iterates. In Figure~\ref{fig:distance} we reported the evolution of the $\ell_2$-distance to the unique EFPE of the game. We can see that Algorithm~\ref{alg:OOMD} outperforms all the other benchmark in terms of this metric. Indeed it achieves a distance more then one order of magnitude smaller then the second best algorithms that are OOMD and CFR. This further corroborates the theoretical findings of Section~\ref{sec:methods}.
Finally we can see that OOMD$(\epsilon=0.01)$ and OOMD$(\epsilon=0.01)$ are ordered as expected, with the one instantiated with the lowest, fixed $\epsilon$ gets closer to $\zvec^\star$.

\end{document}